\documentclass[journal,draftcls,onecolumn,12pt,twoside]{IEEEtranTCOM}
\normalsize

\usepackage[noadjust]{cite}

\ifCLASSINFOpdf

\else
 
\fi

\usepackage{algorithm}
\usepackage{algpseudocode}
\usepackage{textcomp}
\usepackage{pgfplots}
\pgfplotsset{compat=1.18}
\usepackage{float}
\usepackage[compact]{titlesec}
\usepackage{relsize}
\usepackage{booktabs, longtable, multirow, supertabular, tabularx}

    \newcolumntype{P}[1]{>{\centering\hspace{0pt}\arraybackslash}p{#1}}
    \newcolumntype{M}[1]{>{\centering\hspace{0pt}\arraybackslash}m{#1}}
    \newcolumntype{L}{>{\centering\arraybackslash}m{3cm}}

\usepackage{cite}
\def\BibTeX{{\rm B\kern-.05em{\sc i\kern-.025em b}\kern-.08em
    T\kern-.1667em\lower.7ex\hbox{E}\kern-.125emX}}
\usepackage{textcomp}
\let\labelindent\relax
\usepackage{enumitem}
\setlist{nolistsep,leftmargin=.6cm}
\usepackage{bbm}
\usepackage{amsmath, amssymb, amsthm}
\usepackage{psfrag}
\usepackage{graphicx}
\usepackage{tikz}
\usetikzlibrary{arrows.meta,
                fit,
                positioning,
                quotes}

\usepackage{cuted}
\usepackage[font=small]{caption}
\usepackage{pbox}
\usepackage{filecontents}
 
\begin{document}
\newtheorem{proposition}{Proposition}
\newtheorem{definition}{Definition}
\newtheorem{lemma}{Lemma}
\newtheorem*{theorem*}{Theorem}
\newtheorem{theorem}{Theorem}
\newtheorem{corollary}{Corollary}
\newtheorem{assumption}{Assumption}
\newtheorem{claim}{Claim}
%
\title{Towards Efficient Device Identification in Massive Random Access: A Multi-stage Approach}

\author{Jyotish~Robin,~\IEEEmembership{Graduate~Student~Member,~IEEE,}
        Elza~Erkip,~\IEEEmembership{Fellow,~IEEE}
\thanks{Jyotish Robin and Elza Erkip are with the Department of Electrical and
Computer Engineering, NYU Tandon School of Engineering, 6 MetroTech Center, Brooklyn, NY 11201 USA. (e-mail: jyotish.robin@nyu.edu). }
}

\allowdisplaybreaks
\markboth{IEEE Transactions on Communications}%
{Submitted paper}

\maketitle

\vspace{-1cm}
\begin{abstract}
Efficient and low-latency wireless connectivity between the base station (BS) and a sparse set of sporadically  active devices from a massive number of devices is crucial for random access in emerging massive machine-type communications (mMTC). This paper addresses the challenge of identifying active devices while meeting stringent access delay and reliability constraints in mMTC environments. A novel multi-stage active device identification framework is proposed where we aim to  refine a  partial estimate of the active device set  using feedback and hypothesis testing across multiple stages eventually
leading to an exact recovery of active devices after the final stage of processing. In our proposed approach, active devices independently transmit binary preambles during each stage, leveraging feedback signals from the BS, whereas the BS employs a non-coherent binary energy detection. The minimum user identification cost associated with our multi-stage non-coherent active device identification framework with feedback, in terms of the required number of channel-uses, is quantified  using information-theoretic techniques in the asymptotic regime of total number of devices $\ell$ when the number of active devices $k$ scales as $k=\Theta(1)$. Practical implementations of our multi-stage active device identification schemes, leveraging Belief Propagation (BP) techniques, are also presented and evaluated. Simulation results show that our multi-stage BP strategies exhibit superior performance over single-stage strategies, even when considering overhead costs associated with feedback and hypothesis testing.

\end{abstract}


%
\IEEEpeerreviewmaketitle

\section{Introduction}
Next generation of wireless networks strive to provide massive connectivity, facilitating new applications involving massive machine-to-machine  type communications such as in Internet of Things (IoT). These applications entail the deployment of numerous wireless sensors to monitor diverse connections, spanning from wearable biomedical devices to intelligent industrial systems \cite{8454392}. Despite the massive scale of machine-type connections, only a small subset of devices are active at any given time frame as they are triggered by external events \cite{8264818, robin2021sparse}. Thus, efficient and low-latency wireless connectivity between the base station (BS) and a massive number of sporadically active devices is a key service requirement in most of the emerging IoT applications. The conventional approach of allocating dedicated transmission resources to all users overwhelms the available system resources, necessitating prompt detection of active users by the BS to decode their data with minimal latency.

Traditional wireless communication systems typically employ a grant-based strategy for random access, such as ALOHA and its variants, wherein each active device utilizes a preamble sequence to indicate its activity status to the BS \cite{8187056}. However, grant-based strategies are often a suboptimal choice in massive random access scenarios since the high volume of devices results in frequent collisions. Consequently, excessive retransmissions occur, violating the low-latency requirements inherent in IoT systems.
To alleviate the inefficiencies associated with grant-based methods, several recent works have proposed  grant-free (GF) random access schemes eliminating the need for explicit grants from the BS \cite{8454392,9154288, 9537931,8264818,7973146}. The GF schemes offer a more scalable and efficient solution for IoT networks with massive connectivity. The first stage of typical GF transmission process is a user identification phase where the active user transmits a preamble sequence which is utilized by the BS for identifying active users. The sparse nature of the active device set renders the compressed sensing  techniques suitable for active device identification \cite{9060999}. 

Several recent studies \cite{8734871,7952810,8454392,5695122,7282735}  make use of the sparse and sporadic activity pattern of devices to apply compressed sensing techniques like approximate message passing  \cite{8734871,7952810,8454392,5695122} and sparse graph codes \cite{7282735} for identifying the set of active devices. Our prior work in \cite{10198448} considered the active device identification problem with constraints typical to secure, energy-limited mMTC networks. We proposed a single-stage active device identification protocol in which active devices jointly transmit their independent non-orthogonal On-Off Keying (OOK)-based preambles, while the BS employs a non-coherent energy detection strategy to identify these active devices. We noted that the framework outlined above, aimed at recovering the sparse set of active devices using binary measurements, can be regarded as the well-known  group testing (GT) problem \cite{8926588}. We also analyzed various practical schemes  for active device identification in the non-coherent OOK setting using Noisy-Combinatorial Orthogonal Matching Pursuit (N-COMP) and Belief Propagation (BP) techniques \cite{6120391,5169989}. 

One key observation in \cite{10198448} was that our proposed BP based strategy can significantly reduce the number of channel-uses required for user identification when operating under a partial recovery criterion, which allows for a small number of misdetections and false positives. Motivated by this, in this paper,  we explore a multi-stage approach for the exact recovery of active devices. In each stage, partial estimates are computed and successively refined, leading to exact recovery in the final stage. In addition to benefiting from the reduced number of channel-uses required for partial recovery, we also enhance the sparsity of the problem in the further stages. This is expected to improve the overall rate of active device identification since GT decoding is known to have superior performance as the sparsity level of the problem increases \cite{8926588}. As an example, let's consider a two-stage approach with $\ell = 1000$ devices and $k = 20$ active devices as illustrated in Fig. \ref{fig:partailrec2}. In the first stage, our objective is to achieve a partial recovery of, for instance, 90$\%$ of the active devices. Thereafter, in the second stage, we aim for an exact recovery of the remaining 10$\%$ of active devices as well as a validation of the activity status of the $90\%$ that have already been identified. Therefore, as we argue in Section III, our two-stage strategy decomposes the sparse active set identification problem with $\ell = 1000, k = 20$ into a denser problem with $\ell = 20, k = 18$ and a sparser problem with $\ell = 980, k = 2$. The denser problem can be efficiently handled by individual hypothesis testing whereas a GT  based decoder is suitable for the sparser problem. 

In order to facilitate this multi-stage active device identification, the BS needs to provide feedback to the users after each stage to indicate if they are classified as active or inactive. The devices rely on this classification status to determine if they belong to the denser or sparser problem in the next stage. Feedback channels exist in typical modern wireless communication systems such as LTE and 5G, facilitating efficient communication between the transmitter and receiver by providing necessary information for channel estimation, scheduling, and adaptive modulation and coding \cite{fettweis20145g}. In these systems, feedback is typically transmitted over dedicated channels, with the number of channel uses varying based on the specific communication protocol and system configuration. For instance, in 5G systems, the number of feedback channel-uses can range from tens to potentially hundreds, especially in scenarios involving advanced MIMO and beamforming techniques \cite{9031236}. In our proposed active device identification scheme, we show in Section IV-C that the feedback required is minimal, allowing for seamless integration into existing feedback channels without significant overhead. In Section V, we demonstrate that, in practice, our multi-stage active device identification framework can lead to a lower overall number of channel-uses, or equivalently, reduced resource utilization.

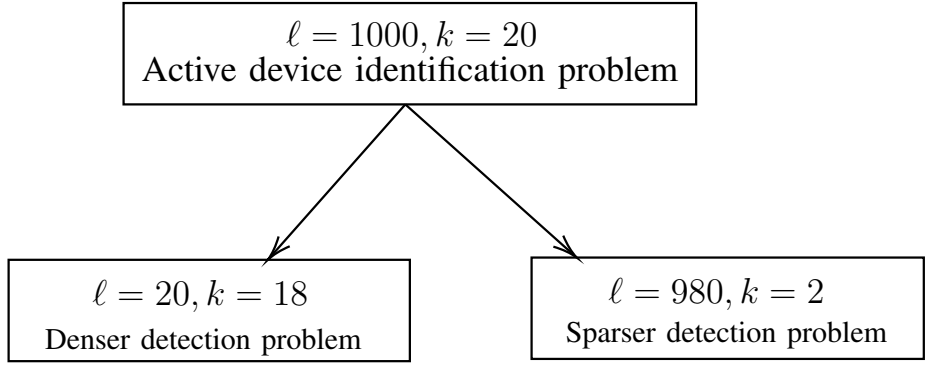
\begin{figure}   
	\centering
	\resizebox{5.0in}{!}{

\tikzset{every picture/.style={line width=0.75pt}} 

\begin{tikzpicture}[x=0.75pt,y=0.75pt,yscale=-1,xscale=1]

\draw   (171.33,36) -- (425.33,36) -- (425.33,81) -- (171.33,81) -- cycle ;
\draw    (296.33,81) -- (236.66,148.5) ;
\draw [shift={(235.33,150)}, rotate = 311.48] [color={rgb, 255:red, 0; green, 0; blue, 0 }  ][line width=0.75]    (10.93,-3.29) .. controls (6.95,-1.4) and (3.31,-0.3) .. (0,0) .. controls (3.31,0.3) and (6.95,1.4) .. (10.93,3.29)   ;
\draw   (120.33,150) -- (298.33,150) -- (298.33,195) -- (120.33,195) -- cycle ;
\draw    (296.33,81) -- (370.84,147.67) ;
\draw [shift={(372.33,149)}, rotate = 221.82] [color={rgb, 255:red, 0; green, 0; blue, 0 }  ][line width=0.75]    (10.93,-3.29) .. controls (6.95,-1.4) and (3.31,-0.3) .. (0,0) .. controls (3.31,0.3) and (6.95,1.4) .. (10.93,3.29)   ;
\draw   (352.33,149) -- (526.33,149) -- (526.33,194) -- (352.33,194) -- cycle ;

\draw (242,42.4) node [anchor=north west][inner sep=0.75pt]    {$\ell =1000,k=20$};
\draw (178,58.4) node [anchor=north west][inner sep=0.75pt]    {Active device identification problem};
\draw (156,157.4) node [anchor=north west][inner sep=0.75pt]    {$\ell =20,k=18$};
\draw (385,156.4) node [anchor=north west][inner sep=0.75pt]    {$\ell =980,k=2$};

\draw (366,176) node [anchor=north west][inner sep=0.75pt]   [align=left] {{\small Sparser detection problem}};
\draw (135,179) node [anchor=north west][inner sep=0.75pt]   [align=left] {{\small Denser detection problem}};

\end{tikzpicture}
}
	\setlength{\belowcaptionskip}{-12pt}
 \setlength{\belowcaptionskip}{-20pt} 
	\caption{Sample illustration of multi-stage active device identification for $\ell =1000$ devices and $k =20 $ active devices. }
	\label{fig:partailrec2}
	\end{figure}

The main contributions of this paper can be summarized as follows:
\begin{enumerate}[label=(\roman*)]
\item 
We propose a novel multi-stage non-coherent activity detection strategy for  mMTC networks in which active devices jointly transmit their OOK-modulated preambles in each stage based on the feedback signal received from BS and the BS uses a non-coherent energy detection strategy to identify the active devices.

 \item We utilize information-theoretic techniques to quantify the minimum user identification cost for our multi-stage  non-coherent strategy with feedback, which represents the minimum number of channel-uses required to identify active users \cite{9517965,7852531}. We show that, in the limit of a large number of users, the fundamental bound on user identification cost remains invariant irrespective of the number of stages.

\item We present an efficient strategy for constructing the feedback signal utilized by the BS to inform each device in the partially estimated set that they are classified as active. Additionally, we analyze the number of channel-uses required to provide this feedback.

\item \sloppy We present and study several practical schemes for multi-stage active device identification in the non-coherent OOK setting based on BP techniques \cite{6120391,5169989}. We also show that the performance gap observed between these practical strategies and our theoretical characterization of the minimum user identification cost is small. Even though (ii) suggests that there is no gain from multi-stage strategies in the asymptotic setting, when the number of users is finite, we show that multi-stage BP strategies can outperform single-stage BP strategies even when the feedback cost is incorporated.

\end{enumerate}

The remainder of this paper is structured as follows. In Section II, we define the system model for the problem of  active device identification in an $m$-stage non-coherent $(\ell,k)$-MnAC system with $\ell$ total devices and $k$ active devices as well as establish its equivalence to a GT problem. In addition, we formally define the key metric of interest, viz, the minimum user identification cost. In Section III, we derive the minimum user identification cost for the $m$-stage non-coherent $(\ell,k)-$MnAC by viewing active device identification as a decoding problem in an equivalent point-to-point communication channel with feedback. In Section IV, we present a  multi-stage active device identification protocol in which the BS refines its estimate of active device set through multiple stages of signal processing and feedback. In Section V, we numerically evaluate several practical strategies for multi-stage active device identification and compare them against the single-stage strategies as well as theoretical bounds. In Section VI, we conclude the paper.

\section{System model}~\label{sec:sysmodl} 
\vspace{-0.9cm}
\subsection{Network model}

Consider an mMTC network comprising of $\ell$ users denoted by the set $\mathcal{D}=\{1,2,\ldots,\ell\}$. Among these $\ell$ users, only a subset of size $k$ is active and wishes to access the channel. It is assumed that the value of $k$ is periodically estimated at the base station (BS) \textit{a-priori} \cite{4085381}. In typical massive random access applications, the value of $\ell$ ranges from approximately $10^4$ to $10^6$, while $k$ typically falls in the range of 10 to 100 \cite{8264818,9564037}. Taking into account sparse user activity, we pay particular attention to the asymptotic regime where $\ell$ can be unbounded while the number of active users scales as $k=\Theta(\ell^{\alpha});$  $0 \leq \alpha <1$. This is consistent with the combinatorial version of the original Many-access channel (MnAC) model \cite{7852531} introduced in our prior work in \cite{10198448}.
The active user set denoted as $\mathcal{A}=\{a_1,a_2,\ldots, a_k\}$, is assumed to be uniform random among all possible ${\ell \choose k} $ subsets of size $k$ from the set $\mathcal{D}$, and is unknown at the BS. The vector $\boldsymbol{\beta}=(\beta_{1},\ldots,\beta_{\ell})$ represents the activity status vector such  that $\beta_{i}=1$, if $i^{th}$ user is active; 0 otherwise. The objective of the BS is to estimate the active user set $\mathcal{A}$ or equivalently $\boldsymbol{\beta}$,  in a timely manner without overwhelming the available network resources for channel access.

 We consider a multi-stage adaptive framework for active device identification in which the  BS can provide feedback to the devices after each stage through a feedback channel and the devices use this feedback to decide on their transmission during the subsequent stage. Specifically, we consider an $m$-stage approach as follows. 
 
 For every stage $j \in \{1,\ldots,m\},$  each user $i \in \mathcal{D}$ employs an On-Off Keying (OOK)-modulated binary preamble  $\textbf{X}^{i,j} \triangleq (X_1^{i,j}, \ldots, X_{n_j}^{i,j}) \in \{0, \sqrt{P}\}^{n_j}$ designed using a Bernoulli random variable $Bern(q_{t}^{i,j}), \forall t \in \{1,\ldots, n_j\}. $  Here $q_{t}^{i,j}$ denotes the probability of `On' symbols for user $i$ during channel-use $t$ of stage $j$. We use $\boldsymbol{q}^j:=\{q_{t}^{i,j}: i \in \mathcal{D}, t \in \{1,\ldots, n_j\}\}$ to denote the set of sampling probabilities for stage $j$. Here, ${P}$ indicates the power level associated with the `On' symbols. 
The Bernoulli RVs are independent across channel-uses $t =\{1,\ldots,n_j\}$, stages $j =\{1,\ldots,m\}$  and users $i =\{1,\ldots,\ell\}$. Note that inactive users $ i \in \mathcal{D}\setminus \mathcal{A}$ use $q_{t}^{i,j}=0$ in order to remain silent by employing a sequence of `Off' symbols of duration $n_j$. We use $\boldsymbol{\vec{0}}_{n_j}$ to denote this all-zero preamble. Once generated, the preambles  $\textbf{X}^{i,j}$'s are fixed and are known to the BS. 

During stage $j$,  each device $i \in \mathcal{D}$ synchronously transmits their preamble $\textbf{X}^{i,j}, \forall j \in \{1,\ldots, m\}$.  Let $\overline{\boldsymbol{X}}{(j)}$ denote the $\ell \times n_j$ transmitted preamble matrix where the $i^{th}$ row is $\textbf{X}^{i,j}$ .  We use ${\boldsymbol{X}_t}{(j)}$ to denote the $t^{th}$ column of $\overline{\boldsymbol{X}}{(j)}$ representing the $t^{th}$ symbols in the preambles of the $\ell$ users during stage $j$. i.e., ${\boldsymbol{X}_t}{(j)} =(X^{1,j}_{t}, \ldots X^{\ell,j}_{t})^T, \forall t \in \{1,\ldots,n_j\}$. In addition, we use ${\tilde{\boldsymbol{X}}_t}{(j)}$ to denote the subset of ${\boldsymbol{X}_t}{(j)} $ corresponding to the $k$ active users. i.e., ${\tilde{\boldsymbol{X}}_t}{(j)} =(X^{a_1,j}_{t}, \ldots X^{a_k,j}_{t})^T, \forall t \in \{1,\ldots,n_j\}$. The symbol received during the $t^{th}$ channel-use  of stage $j,\forall j\in \{1,\ldots,m\}$ is given by
\begin{equation}
    S_t^j=(\boldsymbol{\beta} \circ \mathbf{h}^j_t) \cdot {\boldsymbol{X}_t}{(j)}+W^j_t, \forall t \in\{1,\ldots,n_j\}.
    \label{qqq}
\end{equation} Here, $\textbf{h}^j_t=(h^{1,j}_t,\ldots,h^{\ell,j}_{t})$ is the vector of channel coefficients of the $\ell$ users during the $t^{th}$ channel-use of stage $j$  where $h_t^{i,j}$'s are   i.i.d  $ \mathcal{C} \mathcal{N}\left(0,\sigma^{2}\right), \forall i \in \mathcal{D}, \forall t \in \{1,2,\ldots n_j\}, \forall j \in \{1,\ldots, m\} $. The AWGN noise vector $\textbf{W}^j=(W^j_1,\ldots, W^j_{n_j})$ is composed of i.i.d entries,  $W^j_t  \sim \mathcal{C} \mathcal{N}\left(0, \sigma_w^{2}\right)$ and is independent of $\textbf{h}^j_t$ $\forall t\in \{1,\ldots,n_j\}, j \in \{1,\ldots, m\}$. We use $\boldsymbol{\beta} \circ \mathbf{h}^j_t$ to denote the element-wise product between the vectors $\boldsymbol{\beta}$ and $\mathbf{h}^j_t$. Note that (\ref{qqq}) explicitly shows that only active users contribute to $S_t^j$. 

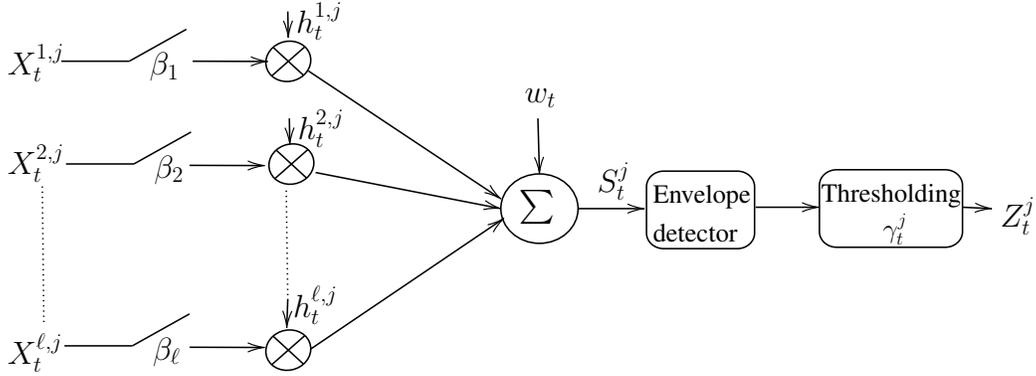
\begin{figure}   
	\centering
	\resizebox{5.5in}{!}{
\tikzset{every picture/.style={line width=1pt}} 

\begin{tikzpicture}[x=1pt,y=0.9pt,yscale=-1,xscale=1]

\draw  [dash pattern={on 0.84pt off 2.51pt}]  (55.67,152.17) -- (56.08,193.14) -- (56.67,251.17) ;
\draw    (153,62) -- (199.67,62.16) ;
\draw [shift={(201.67,62.17)}, rotate = 180.2] [color={rgb, 255:red, 0; green, 0; blue, 0 }  ][line width=0.75]    (10.93,-3.29) .. controls (6.95,-1.4) and (3.31,-0.3) .. (0,0) .. controls (3.31,0.3) and (6.95,1.4) .. (10.93,3.29)   ;
\draw    (151,137) -- (197.67,137.16) ;
\draw [shift={(199.67,137.17)}, rotate = 180.2] [color={rgb, 255:red, 0; green, 0; blue, 0 }  ][line width=0.75]    (10.93,-3.29) .. controls (6.95,-1.4) and (3.31,-0.3) .. (0,0) .. controls (3.31,0.3) and (6.95,1.4) .. (10.93,3.29)   ;
\draw    (151,268) -- (197.67,268.16) ;
\draw [shift={(199.67,268.17)}, rotate = 180.2] [color={rgb, 255:red, 0; green, 0; blue, 0 }  ][line width=0.75]    (10.93,-3.29) .. controls (6.95,-1.4) and (3.31,-0.3) .. (0,0) .. controls (3.31,0.3) and (6.95,1.4) .. (10.93,3.29)   ;
\draw   (225.72,51.59) .. controls (231.16,57.15) and (230.77,66.35) .. (224.85,72.15) .. controls (218.93,77.94) and (209.72,78.14) .. (204.28,72.58) .. controls (198.84,67.02) and (199.23,57.82) .. (205.15,52.02) .. controls (211.07,46.22) and (220.28,46.03) .. (225.72,51.59) -- cycle ; \draw   (225.72,51.59) -- (204.28,72.58) ; \draw   (224.85,72.15) -- (205.15,52.02) ;
\draw   (227.72,125.59) .. controls (233.16,131.15) and (232.77,140.35) .. (226.85,146.15) .. controls (220.93,151.94) and (211.72,152.14) .. (206.28,146.58) .. controls (200.84,141.02) and (201.23,131.82) .. (207.15,126.02) .. controls (213.07,120.22) and (222.28,120.03) .. (227.72,125.59) -- cycle ; \draw   (227.72,125.59) -- (206.28,146.58) ; \draw   (226.85,146.15) -- (207.15,126.02) ;
\draw   (225.72,259.59) .. controls (231.16,265.15) and (230.77,274.35) .. (224.85,280.15) .. controls (218.93,285.94) and (209.72,286.14) .. (204.28,280.58) .. controls (198.84,275.02) and (199.23,265.82) .. (205.15,260.02) .. controls (211.07,254.22) and (220.28,254.03) .. (225.72,259.59) -- cycle ; \draw   (225.72,259.59) -- (204.28,280.58) ; \draw   (224.85,280.15) -- (205.15,260.02) ;
\draw    (214.67,233.17) -- (214.06,252) ;
\draw [shift={(214,254)}, rotate = 271.83] [color={rgb, 255:red, 0; green, 0; blue, 0 }  ][line width=0.75]    (10.93,-3.29) .. controls (6.95,-1.4) and (3.31,-0.3) .. (0,0) .. controls (3.31,0.3) and (6.95,1.4) .. (10.93,3.29)   ;
\draw    (215.67,103.17) -- (215.96,120) ;
\draw [shift={(216,122)}, rotate = 268.99] [color={rgb, 255:red, 0; green, 0; blue, 0 }  ][line width=0.75]    (10.93,-3.29) .. controls (6.95,-1.4) and (3.31,-0.3) .. (0,0) .. controls (3.31,0.3) and (6.95,1.4) .. (10.93,3.29)   ;
\draw    (214.67,26.17) -- (214.96,43) ;
\draw [shift={(215,45)}, rotate = 268.99] [color={rgb, 255:red, 0; green, 0; blue, 0 }  ][line width=0.75]    (10.93,-3.29) .. controls (6.95,-1.4) and (3.31,-0.3) .. (0,0) .. controls (3.31,0.3) and (6.95,1.4) .. (10.93,3.29)   ;
\draw    (228,69) -- (351.06,159.98) ;
\draw [shift={(352.67,161.17)}, rotate = 216.48] [color={rgb, 255:red, 0; green, 0; blue, 0 }  ][line width=0.75]    (10.93,-3.29) .. controls (6.95,-1.4) and (3.31,-0.3) .. (0,0) .. controls (3.31,0.3) and (6.95,1.4) .. (10.93,3.29)   ;
\draw    (232.67,142.17) -- (350.71,167.74) ;
\draw [shift={(352.67,168.17)}, rotate = 192.23] [color={rgb, 255:red, 0; green, 0; blue, 0 }  ][line width=0.75]    (10.93,-3.29) .. controls (6.95,-1.4) and (3.31,-0.3) .. (0,0) .. controls (3.31,0.3) and (6.95,1.4) .. (10.93,3.29)   ;
\draw    (229.67,269.17) -- (353.07,176.37) ;
\draw [shift={(354.67,175.17)}, rotate = 143.06] [color={rgb, 255:red, 0; green, 0; blue, 0 }  ][line width=0.75]    (10.93,-3.29) .. controls (6.95,-1.4) and (3.31,-0.3) .. (0,0) .. controls (3.31,0.3) and (6.95,1.4) .. (10.93,3.29)   ;
\draw  [dash pattern={on 0.84pt off 2.51pt}]  (213.67,155.17) -- (214.67,233.17) ;
\draw   (352.67,168.17) .. controls (352.67,154.36) and (363.86,143.17) .. (377.67,143.17) .. controls (391.47,143.17) and (402.67,154.36) .. (402.67,168.17) .. controls (402.67,181.97) and (391.47,193.17) .. (377.67,193.17) .. controls (363.86,193.17) and (352.67,181.97) .. (352.67,168.17) -- cycle ;
\draw    (376.67,103.17) -- (377.62,141.17) ;
\draw [shift={(377.67,143.17)}, rotate = 268.57] [color={rgb, 255:red, 0; green, 0; blue, 0 }  ][line width=0.75]    (10.93,-3.29) .. controls (6.95,-1.4) and (3.31,-0.3) .. (0,0) .. controls (3.31,0.3) and (6.95,1.4) .. (10.93,3.29)   ;
\draw    (402.67,168.17) -- (443.67,168.17) ;
\draw [shift={(445.67,168.17)}, rotate = 180] [color={rgb, 255:red, 0; green, 0; blue, 0 }  ][line width=0.75]    (10.93,-3.29) .. controls (6.95,-1.4) and (3.31,-0.3) .. (0,0) .. controls (3.31,0.3) and (6.95,1.4) .. (10.93,3.29)   ;
\draw   (447,152.33) .. controls (447,146.17) and (452,141.17) .. (458.17,141.17) -- (505.83,141.17) .. controls (512,141.17) and (517,146.17) .. (517,152.33) -- (517,185.83) .. controls (517,192) and (512,197) .. (505.83,197) -- (458.17,197) .. controls (452,197) and (447,192) .. (447,185.83) -- cycle ;
\draw    (517.67,167.17) -- (556.67,167.17) ;
\draw [shift={(558.67,167.17)}, rotate = 180] [color={rgb, 255:red, 0; green, 0; blue, 0 }  ][line width=0.75]    (10.93,-3.29) .. controls (6.95,-1.4) and (3.31,-0.3) .. (0,0) .. controls (3.31,0.3) and (6.95,1.4) .. (10.93,3.29)   ;
\draw   (559,151.33) .. controls (559,145.17) and (564,140.17) .. (570.17,140.17) -- (640.5,140.17) .. controls (646.67,140.17) and (651.67,145.17) .. (651.67,151.33) -- (651.67,184.83) .. controls (651.67,191) and (646.67,196) .. (640.5,196) -- (570.17,196) .. controls (564,196) and (559,191) .. (559,184.83) -- cycle ;
\draw    (651.67,167.17) -- (668.67,168.06) ;
\draw [shift={(670.67,168.17)}, rotate = 183.01] [color={rgb, 255:red, 0; green, 0; blue, 0 }  ][line width=0.75]    (10.93,-3.29) .. controls (6.95,-1.4) and (3.31,-0.3) .. (0,0) .. controls (3.31,0.3) and (6.95,1.4) .. (10.93,3.29)   ;
\draw    (68,62) -- (112,62) ;
\draw    (112,62) -- (149,39) ;

\draw    (71,136) -- (115,136) ;
\draw    (115,136) -- (152,113) ;

\draw    (70,267) -- (114,267) ;
\draw    (114,267) -- (151,244) ;

\draw (32,48.4) node [anchor=north west][inner sep=0.75pt]   [font=\LARGE]{$X_{t}^{1,j}$};
\draw (32,120.4) node [anchor=north west][inner sep=0.75pt]   [font=\LARGE]{$X_{t}^{2,j}$};
\draw (32,256.4) node [anchor=north west][inner sep=0.75pt]   [font=\LARGE] {$X_{t}^{\ell,j}$};
\draw (219,18.4) node [anchor=north west][inner sep=0.75pt]  [font=\LARGE] {$h_t^{1,j}$};
\draw (220,95.4) node [anchor=north west][inner sep=0.75pt]  [font=\LARGE]  {$h_t^{2,j}$};
\draw (217,222.4) node [anchor=north west][inner sep=0.75pt] [font=\LARGE]  {$h_t^{\ell,j}$};
\draw (363,153.4) node [anchor=north west][inner sep=0.75pt]  [font=\huge]  {$\sum $};
\draw (367,80.4) node [anchor=north west][inner sep=0.75pt]  [font=\LARGE] {$w_{t}$};
\draw (414,133.4) node [anchor=north west][inner sep=0.75pt]  [font=\LARGE]  {$S_{t}^j$};
\draw (559,149) node [anchor=north west][inner sep=0.75pt]  [font=\Large][align=left] {Thresholding};
\draw (599,170) node [anchor=north west][inner sep=0.75pt]  [font=\Large][align=left] {$\gamma_t^j$};
\draw (450,149.33) node [anchor=north west][inner sep=0.75pt]  [font=\Large] [align=left] {Envelope \\[-10pt] detector};
\draw (675,158.4) node [anchor=north west][inner sep=0.75pt][font=\LARGE]  {$Z_{t}^j$};
\draw (124,54.4) node [anchor=north west][inner sep=0.75pt]  [font=\LARGE]{$\beta _{1}$};
\draw (127,128.4) node [anchor=north west][inner sep=0.75pt]   [font=\LARGE]{$\beta _{2}$};
\draw (126,259.4) node [anchor=north west][inner sep=0.75pt] [font=\LARGE]  {$\beta _{\ell}$};

\end{tikzpicture}

}
	\setlength{\belowcaptionskip}{-22pt}
	\caption{Non-coherent $(\ell,k)$-Many Access Channel at channel-use $t$ of stage $j$; $j \in \{1,\ldots,m\};t \in \{1,\ldots,n_j(\ell)\}$. }
	\label{fig:qtzn}
	\end{figure}

We assume that due to the fast fading nature of the channel which could arise, for example, in secure mMTC networks \cite{9360811}, there is no channel state information (CSI) available at the BS and hence the BS employs non-coherent detection \cite{6120373} as shown in Fig. \ref{fig:qtzn}. During the $t^{th}$ channel use, the envelope detector processes  $S^j_t$ to obtain $|S^j_t|$, the envelope of $S^j_t$.  Thresholding  after envelope detection produces a binary output $Z^j_{t}$ such that
\begin{equation}
  Z_{t}^j=1 \text { if } |S^j_t|^2>\gamma_t^j ; \text { else } Z_{t}^j=0.  
  \label{threshold}
\end{equation} Here, $\gamma_t^j $ represents a pre-determined threshold parameter that may depend on the stage and channel-use indices. We use $\boldsymbol{Z}^{j}\triangleq (Z^j_1,\ldots Z^j_{n_j})$ to denote the output vector at stage $j$. We use the terminology \textit{non-coherent $(\ell, k)-$MnAC} to refer to the channel model described above where devices are limited by OOK based transmissions and the BS is limited by non-coherent energy detection.

In contrast to our prior work in \cite{10198448} where we had a single stage ($m=1$), here we consider a model in which the BS broadcasts a feedback signal, $\psi_j(\boldsymbol{Z^1},\ldots,\boldsymbol{Z^{j-1}})$, to all devices over a  feedback channel after each stage $j$. For brevity, we will refer to this feedback signal as $\psi_j$, which is a function of all previous channel outputs $\{\boldsymbol{Z^1},\ldots,\boldsymbol{Z^{j-1}}\}.$ This feedback enables each device $i$, instead of solely using its activity status $\beta_i$, to use the feedback from previous stages $\{\psi_1,\ldots,\psi_{j-1}\}$ for  preamble generation at stage $j$. Hence, $\boldsymbol{X}^{i,j}$ is a function of $\{\psi_1,\ldots,\psi_{j-1}\}$  and  $\beta_i$ as illustrated in Fig. \ref{fig:fbch}. Note that there is no co-operation among devices during preamble generation. i.e., preambles are generated independently. The objective of the BS is to  determine the active user set $\mathcal{A}$ or, equivalently, infer the activity status vector $\boldsymbol{\beta}$ after $m$ stages using the output vectors  $\{\textbf{Z}^1,\ldots,\textbf{Z}^m\}$.

\begin{figure}   
	\centering
	\resizebox{6.2in}{!}{%
\tikzset{every picture/.style={line width=0.85pt}} 

\begin{tikzpicture}[x=0.75pt,y=0.75pt,yscale=-1,xscale=1]

\draw   (38.94,152.64) -- (179.89,152.64) -- (179.89,214.33) -- (38.94,214.33) -- cycle ;
\draw    (84,84.27) -- (471.48,85.34) ;
\draw    (454.34,173.96) -- (498.69,173.96) ;
\draw    (471.2,173.27) -- (471.47,87.34) ;
\draw [shift={(471.48,84.34)}, rotate = 90.23] [color={rgb, 255:red, 0; green, 0; blue, 0 }  ][line width=0.75]    (10.93,-3.29) .. controls (6.95,-1.4) and (3.31,-0.3) .. (0,0) .. controls (3.31,0.3) and (6.95,1.4) .. (10.93,3.29)   ;
\draw    (84,84.27) -- (84.38,150.67) ;
\draw [shift={(84.4,152.67)}, rotate = 269.53] [color={rgb, 255:red, 0; green, 0; blue, 0 }  ][line width=0.75]    (10.93,-3.29) .. controls (6.95,-1.4) and (3.31,-0.3) .. (0,0) .. controls (3.31,0.3) and (6.95,1.4) .. (10.93,3.29)   ;
\draw    (179.9,172.67) -- (356,172.17) ;
\draw   (356.62,147.49) -- (456.21,147.49) -- (456.21,204.09) -- (356.62,204.09) -- cycle ;
\draw   (498.56,146.57) -- (585.21,146.57) -- (585.21,202.26) -- (498.56,202.26) -- cycle ;
\draw    (13.3,163.53) -- (36.31,163.53) ;
\draw [shift={(38.31,163.53)}, rotate = 180] [color={rgb, 255:red, 0; green, 0; blue, 0 }  ][line width=0.75]    (10.93,-3.29) .. controls (6.95,-1.4) and (3.31,-0.3) .. (0,0) .. controls (3.31,0.3) and (6.95,1.4) .. (10.93,3.29)   ;
\draw    (13.3,175.4) -- (36.31,175.4) ;
\draw [shift={(38.31,175.4)}, rotate = 180] [color={rgb, 255:red, 0; green, 0; blue, 0 }  ][line width=0.75]    (10.93,-3.29) .. controls (6.95,-1.4) and (3.31,-0.3) .. (0,0) .. controls (3.31,0.3) and (6.95,1.4) .. (10.93,3.29)   ;
\draw    (15.04,203.7) -- (38.05,203.7) ;
\draw [shift={(40.05,203.7)}, rotate = 180] [color={rgb, 255:red, 0; green, 0; blue, 0 }  ][line width=0.75]    (10.93,-3.29) .. controls (6.95,-1.4) and (3.31,-0.3) .. (0,0) .. controls (3.31,0.3) and (6.95,1.4) .. (10.93,3.29)   ;
\draw [line width=1.5]  [dash pattern={on 1.69pt off 2.76pt}]  (22.02,181.79) -- (21.73,200.04) ;
\draw    (585.95,175.79) -- (690.33,176.82) ;
\draw [shift={(690.33,176.83)}, rotate = 180.43] [color={rgb, 255:red, 0; green, 0; blue, 0 }  ][line width=0.75]    (10.93,-3.29) .. controls (6.95,-1.4) and (3.31,-0.3) .. (0,0) .. controls (3.31,0.3) and (6.95,1.4) .. (10.93,3.29)   ;

\draw (513.66,156.87) node [anchor=north west][inner sep=0.75pt]   [align=left] {{Decode}};
\draw (499.66,176.87) node [anchor=north west][inner sep=0.75pt][font=\small]   [align=left] {\text{after $m$ stages}};
\draw (295.66,209.87) node [anchor=north west][inner sep=0.75pt]   [align=left][font=\small] {$m$-stages: $j=\{1,\ldots m\}$ (See Fig.1 for each stage)};
\draw (38.79,162.45) node [anchor=north west][inner sep=0.75pt]  [font=\normalsize] [align=left] {\text{Preamble }};
\draw (106.79,162.45) node [anchor=north west][inner sep=0.75pt]  [font=\normalsize] [align=left] {\text{generation}};
\draw (48.19,187.45) node [anchor=north west][inner sep=0.75pt]  [font=\normalsize] [align=left] {\text{ for each user $i$}};
\draw (176.08,88.8) node [anchor=north west][inner sep=0.75pt]   [align=left] {feedback link after stage $j$};
\draw (89.08,97.8) node [anchor=north west][inner sep=0.75pt]   [align=left] {$\psi_j$};
\draw (356.87,155.86) node [anchor=north west][inner sep=0.75pt]  [font=\normalsize] [align=left] {\begin{minipage}[lt]{67.73pt}\setlength\topsep{0pt}
\begin{center}
Non-coherent \\[-5pt]$\displaystyle ( \ell ,k) -$MnAC 
\end{center}

\end{minipage}};
\draw (120.02,143.64) node [anchor=north west][inner sep=0.75pt]   [align=left] {$ $};
\draw (-3.59,151.49) node [anchor=north west][inner sep=0.75pt]  [font=\small]  {$\beta _{1}$};
\draw (-4.46,168.84) node [anchor=north west][inner sep=0.75pt]  [font=\small]  {$\beta _{2}$};
\draw (-4.46,197.14) node [anchor=north west][inner sep=0.75pt]  [font=\small]  {$\beta _{\ell }$};
\draw (183.06,178.49) node [anchor=north west][inner sep=0.75pt]  [font=\small]  {${\boldsymbol{X}^{i,j}}: f(\beta _{i} ,\psi_{1} ,\dotsc ,\psi_{j-1} )$};
\draw (61.06,218.49) node [anchor=north west][inner sep=0.75pt]  [font=\small]  {$i \in \{1,\ldots \ell\}$};
\draw (474.42,149.64) node [anchor=north west][inner sep=0.75pt]  [font=\normalsize]  {$\mathbf{Z}^{j}$};
\draw (585.93,149.12) node [anchor=north west][inner sep=0.75pt]  [font=\small]  {$\hat{\boldsymbol{\beta }} =\{\hat{\beta }_{1},\dotsc ,\hat{\beta }_{\ell } \}$};
\end{tikzpicture}
}
\caption{Non-coherent $(\ell,k)-$MnAC with stagewise feedback. Index $i =\{1,\ldots \ell\}$ denotes users; $j =\{1,\ldots m\}$ denotes stages where each stage $j$ lasts for $n_j$ channel-uses. Each device preamble is generated independently.}
	\label{fig:fbch}
	\end{figure}
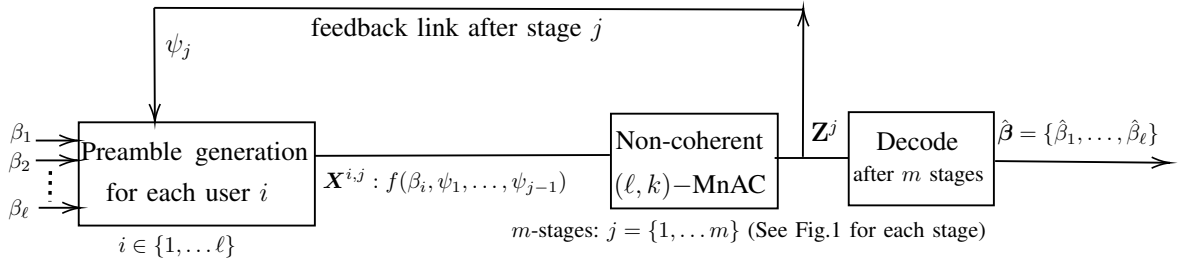

 Consistent with our prior work  \cite{10198448},  we use the following definitions modified to take into account the multi-stage activity detection:

\begin{definition}
\textbf{ (Multi-stage Activity Detection)}: Given $m$ stages and the corresponding set of preambles $\{\mathbf{X}^{i,j}: i \in \mathcal{D}, j \in \{1,\ldots,m\}\}$, an activity detection function   denoted by $\hat{\boldsymbol{\beta}}: \underbrace{\{0,1\}^{n_1} \times \{0,1\}^{n_2} \times \ldots \times \{0,1\}^{n_m}}_{m \text{ terms}} \to \{0,1\}
^{\ell}$  is a deterministic rule that maps the set of binary  channel output vectors $\{\textbf{Z}^1, \textbf{Z}^2, \ldots \textbf{Z}^m\}$ to an estimate $\hat{\boldsymbol{\beta}}$ of the  activity status vector such that the Hamming weight of $\hat{\boldsymbol{\beta}}$ is $k$.
\end{definition}

\begin{definition}
\textbf{(Probability of Erroneous Identification)}: For an activity detection function  $\hat{\boldsymbol{\beta}}$, probability of erroneous identification  $\mathbb{P}_e^{(\ell)}$ is defined as \begin{equation}
     \mathbb{P}_e^{(\ell)}:=\frac{1}{{\ell \choose k}} \sum_{\boldsymbol{\beta}:\sum_{i=1}^{\ell}{\beta_i}=k} \mathbb{P}(\hat{\boldsymbol{\beta}} \neq \boldsymbol{\beta}). \label{errid}
 \end{equation} 
\end{definition}

\begin{definition}
\textbf{ ($\eta \% -$ Partial Recovery)}: For a true active set $\mathcal{A}$ and an estimated active set $\hat{\mathcal{A}}$,  consider an error event $E_1$  defined as 
\begin{equation} \vspace{-0.2cm}
    E_1 := \left\{\left| \mathcal{A} \setminus \hat{\mathcal{A}}\right| >  k\left(1-\frac{\eta}{100}\right) \right\}. \label{errev}
\end{equation}  We have $\eta \%$-partial recovery, if \vspace{-0.4cm} \begin{equation} 
    \mathbb{P}_{e,\eta}^{(\ell)}:= P( E_1)  \rightarrow 0 \text { as } \ell \rightarrow \infty, \nonumber
\end{equation} 
where $\mathbb{P}_{e,\eta}^{(\ell)}$ denotes the probability of error for $\eta \% -$ partial recovery.  Probability  of successful identification for $\eta \% -$  recovery is defined as $\mathbb{P}_{succ,\eta}^{(\ell)}:= 1 -\mathbb{P}_{e,\eta}^{(\ell)}.$

\end{definition}
\begin{definition}
\textbf{(Multi-stage Minimum User Identification Cost)}: The minimum user identification
cost for $m$-stage non-coherent $(\ell,k)-$MnAC is said to be $n^{m}(\ell):=\sum_{j=1}^{m} n_j(\ell)$ if for every $0 < \epsilon <1$, there exists a set of preambles $\{\boldsymbol{X}^{i,j}: i \in \mathcal{D}, j \in \{1,\ldots,m\}\}$ of total length $n_o = (1 + \epsilon)n^{m}(\ell)$  such that the probability
of erroneous identification $\mathbb{P}_e^{(\ell)}$  tends to 0 as $\ell \rightarrow \infty$, whereas
$\mathbb{P}_e^{(\ell)}$  is strictly bounded away from 0 if
$n_0 = (1 - \epsilon)n^{m}(\ell)$.

\end{definition}

Our goal in this paper is to characterize the $m$-stage minimum user identification cost theoretically and identify practical strategies that get close to those bounds. Note that Def. 4 does not take the feedback cost  into account  and only focuses on the number of channel uses required in the forward link (uplink) from devices to the BS. Hence, without loss of generality, in obtaining theoretical bounds,  we can assume $\psi_j(\boldsymbol{Z^1},\ldots,\boldsymbol{Z^{j-1}}) = \boldsymbol{Z}^{j}, \forall j \in \{1,\ldots, m\}$.  This approach is consistent with the information theoretical analysis for communication channels with feedback \cite{10.5555/1146355}. However, when we devise practical schemes in Section IV, 
 we take the feedback cost into account and present feedback schemes whose overhead is much smaller compared to the required number of channel-uses in the uplink. Furthermore, our numerical results in Section IV-C suggests that even when the feedback cost is incorporated into the user identification cost, practical multi-stage strategies are very effective and can operate close to the theoretical bounds established in Section III.
	\section{minimum user identification cost for multi-stage non-coherent $(\ell, k)-$MnAC	with feedback}\label{sec3}
	
 In this section, we analyze the minimum user identification cost for the multi-stage non-coherent $(\ell,k)-$ MnAC with feedback and provide an exact characterization in the $k=\Theta(1)$ regime.  

 \subsection{Equivalent channel model and one-stage strategies}

 \begin{figure}   
	\centering
	\resizebox{4.7in}{!}{%
\tikzset{every picture/.style={line width=1.4pt}} 

\begin{tikzpicture}[x=.9pt,y=.85pt,yscale=-1,xscale=1.1]

\draw   (59.92,86.98) -- (92.32,86.98) -- (92.32,119.45) -- (59.92,119.45) -- cycle ;
\draw   (59.92,131.79) -- (92.32,131.79) -- (92.32,167.63) -- (59.92,167.63) -- cycle ;
\draw   (59.92,221.39) -- (92.32,221.39) -- (92.32,257.23) -- (59.92,257.23) -- cycle ;
\draw  [dash pattern={on 1.84pt off 2.51pt}]  (76.25,175.25) -- (76.77,216.83) ;
\draw    (92.82,145.23) -- (145.76,145.95) ;
\draw [shift={(147.76,145.97)}, rotate = 180.78] [color={rgb, 255:red, 0; green, 0; blue, 0 }  ][line width=0.75]    (10.93,-3.29) .. controls (6.95,-1.4) and (3.31,-0.3) .. (0,0) .. controls (3.31,0.3) and (6.95,1.4) .. (10.93,3.29)   ;
\draw    (91.83,239.31) -- (144.78,240.03) ;
\draw [shift={(146.78,240.06)}, rotate = 180.78] [color={rgb, 255:red, 0; green, 0; blue, 0 }  ][line width=0.75]    (10.93,-3.29) .. controls (6.95,-1.4) and (3.31,-0.3) .. (0,0) .. controls (3.31,0.3) and (6.95,1.4) .. (10.93,3.29)   ;
\draw    (92.82,107.59) -- (145.76,108.31) ;
\draw [shift={(147.76,108.34)}, rotate = 180.78] [color={rgb, 255:red, 0; green, 0; blue, 0 }  ][line width=0.75]    (10.93,-3.29) .. controls (6.95,-1.4) and (3.31,-0.3) .. (0,0) .. controls (3.31,0.3) and (6.95,1.4) .. (10.93,3.29)   ;
\draw   (155.99,80.56) -- (262.26,80.56) -- (262.26,259.77) -- (155.99,259.77) -- cycle ;
\draw    (263.57,96.84) -- (285.26,96.7) ;
\draw [shift={(287.26,96.69)}, rotate = 179.64] [color={rgb, 255:red, 0; green, 0; blue, 0 }  ][line width=0.75]    (10.93,-3.29) .. controls (6.95,-1.4) and (3.31,-0.3) .. (0,0) .. controls (3.31,0.3) and (6.95,1.4) .. (10.93,3.29)   ;
\draw    (263.57,116.55) -- (285.26,116.42) ;
\draw [shift={(287.26,116.4)}, rotate = 179.64] [color={rgb, 255:red, 0; green, 0; blue, 0 }  ][line width=0.75]    (10.93,-3.29) .. controls (6.95,-1.4) and (3.31,-0.3) .. (0,0) .. controls (3.31,0.3) and (6.95,1.4) .. (10.93,3.29)   ;
\draw    (263.57,250.06) -- (285.26,249.92) ;
\draw [shift={(287.26,249.91)}, rotate = 179.64] [color={rgb, 255:red, 0; green, 0; blue, 0 }  ][line width=0.75]    (10.93,-3.29) .. controls (6.95,-1.4) and (3.31,-0.3) .. (0,0) .. controls (3.31,0.3) and (6.95,1.4) .. (10.93,3.29)   ;
\draw  [dash pattern={on 0.84pt off 2.51pt}]  (310.95,137.91) -- (310.95,147.77) -- (310.95,158.52) ;
\draw  [dash pattern={on 0.84pt off 2.51pt}]  (280.95,50.91) -- (280.95,248.52) ;
\draw   (468.87,104.91) .. controls (468.87,92.53) and (479.92,82.5) .. (493.55,82.5) .. controls (507.18,82.5) and (518.22,92.53) .. (518.22,104.91) .. controls (518.22,117.28) and (507.18,127.31) .. (493.55,127.31) .. controls (479.92,127.31) and (468.87,117.28) .. (468.87,104.91) -- cycle ;
\draw   (470.85,235.73) .. controls (470.85,223.35) and (481.89,213.32) .. (495.52,213.32) .. controls (509.15,213.32) and (520.2,223.35) .. (520.2,235.73) .. controls (520.2,248.1) and (509.15,258.13) .. (495.52,258.13) .. controls (481.89,258.13) and (470.85,248.1) .. (470.85,235.73) -- cycle ;
\draw    (321.81,91.32) -- (473.41,249.36) ;
\draw [shift={(474.79,250.81)}, rotate = 226.19] [color={rgb, 255:red, 0; green, 0; blue, 0 }  ][line width=0.75]    (10.93,-3.29) .. controls (6.95,-1.4) and (3.31,-0.3) .. (0,0) .. controls (3.31,0.3) and (6.95,1.4) .. (10.93,3.29)   ;
\draw    (322.79,118.2) -- (473.29,249.49) ;
\draw [shift={(474.79,250.81)}, rotate = 221.1] [color={rgb, 255:red, 0; green, 0; blue, 0 }  ][line width=0.75]    (10.93,-3.29) .. controls (6.95,-1.4) and (3.31,-0.3) .. (0,0) .. controls (3.31,0.3) and (6.95,1.4) .. (10.93,3.29)   ;
\draw    (324.77,250.81) -- (472.79,250.81) ;
\draw [shift={(474.79,250.81)}, rotate = 180] [color={rgb, 255:red, 0; green, 0; blue, 0 }  ][line width=0.75]    (10.93,-3.29) .. controls (6.95,-1.4) and (3.31,-0.3) .. (0,0) .. controls (3.31,0.3) and (6.95,1.4) .. (10.93,3.29)   ;
\draw    (322.79,118.2) -- (472.82,91.66) ;
\draw [shift={(474.79,91.32)}, rotate = 169.97] [color={rgb, 255:red, 0; green, 0; blue, 0 }  ][line width=0.75]    (10.93,-3.29) .. controls (6.95,-1.4) and (3.31,-0.3) .. (0,0) .. controls (3.31,0.3) and (6.95,1.4) .. (10.93,3.29)   ;
\draw    (321.81,91.32) -- (472.79,91.32) ;
\draw [shift={(474.79,91.32)}, rotate = 180] [color={rgb, 255:red, 0; green, 0; blue, 0 }  ][line width=0.75]    (10.93,-3.29) .. controls (6.95,-1.4) and (3.31,-0.3) .. (0,0) .. controls (3.31,0.3) and (6.95,1.4) .. (10.93,3.29)   ;
\draw    (324.77,250.81) -- (473.42,92.77) ;
\draw [shift={(474.79,91.32)}, rotate = 133.25] [color={rgb, 255:red, 0; green, 0; blue, 0 }  ][line width=0.75]    (10.93,-3.29) .. controls (6.95,-1.4) and (3.31,-0.3) .. (0,0) .. controls (3.31,0.3) and (6.95,1.4) .. (10.93,3.29)   ;
\draw   (301.08,80.56) -- (319.83,80.56) -- (319.83,259.77) -- (301.08,259.77) -- cycle ;
\draw  [dash pattern={on 0.84pt off 2.51pt}]  (310.95,200.63) -- (310.95,210.49) -- (309.96,233.78) ;
\draw    (322.79,171.96) -- (473.03,92.25) ;
\draw [shift={(474.79,91.32)}, rotate = 152.05] [color={rgb, 255:red, 0; green, 0; blue, 0 }  ][line width=0.75]    (10.93,-3.29) .. controls (6.95,-1.4) and (3.31,-0.3) .. (0,0) .. controls (3.31,0.3) and (6.95,1.4) .. (10.93,3.29)   ;
\draw    (322.79,171.96) -- (473.02,249.89) ;
\draw [shift={(474.79,250.81)}, rotate = 207.42] [color={rgb, 255:red, 0; green, 0; blue, 0 }  ][line width=0.75]    (10.93,-3.29) .. controls (6.95,-1.4) and (3.31,-0.3) .. (0,0) .. controls (3.31,0.3) and (6.95,1.4) .. (10.93,3.29)   ;
\draw  [fill={rgb, 255:red, 155; green, 155; blue, 155 }  ,fill opacity=0.23 ] (119.47,103.11) .. controls (119.47,79.36) and (138.72,60.1) .. (162.48,60.1) -- (402.17,60.1) .. controls (425.93,60.1) and (445.18,79.36) .. (445.18,103.11) -- (445.18,232.14) .. controls (445.18,255.89) and (425.93,275.15) .. (402.17,275.15) -- (162.48,275.15) .. controls (138.72,275.15) and (119.47,255.89) .. (119.47,232.14) -- cycle ;

\draw (165,104.4) node [anchor=north west][inner sep=0.75pt]  [font=\huge]  {$\frac{\sum _{i=1}^{k}X_{t}^{a_{i},j}}{\sqrt{P}}$};
\draw (156,189) node [anchor=north west][inner sep=0.75pt]  [font=\large] [align=left] {{\large \ \ \ Hamming  }\\[-10pt] {\large \ \ \ weight $(V)$}\\[-10pt] {\large \  computation}};
\draw (28,32) node [anchor=north west][inner sep=0.75pt]  [font=\small] [align=left] {{\large Channel Input}\\{\large  $\displaystyle \ \ \ \ \ \ \ \ \Tilde{\textbf{X}}_t(j)$}};

\draw (220.53,20.96) node [anchor=north west, inner sep=0.75pt, font=\large] {
    \begin{minipage}{4cm} 
        Hamming weight \\[-12pt]
        $\text{\hspace{1.8cm}}V$
    \end{minipage}
};

\draw (487.53,95.96) node [anchor=north west][inner sep=0.75pt]  [font=\LARGE]  {$0$};
\draw (488.53,225.39) node [anchor=north west][inner sep=0.75pt]  [font=\LARGE]  {$1$};
\draw (447,32) node [anchor=north west][inner sep=0.75pt]  [font=\large] [align=left] {{\large Channel Output}\\[-10pt]{\large $\displaystyle \ \ \ \ \ \ \ \ Z_t^j$}};
\draw (48.2,284.73) node [anchor=north west][inner sep=0.75pt]  [font=\large]  {$active\ user\ set:\ \{a_{1} ,a_{2} ,\dotsc ,a_{k} \}$};
\draw (306.92,168.09) node [anchor=north west][inner sep=0.75pt]    {$i$};
\draw (368.51,123.22) node [anchor=north west][inner sep=0.75pt]  [font=\large,rotate=-332.77]  {$p_v^t(j)$};
\draw (373.93,196.77) node [anchor=north west][inner sep=0.75pt]  [font=\large,rotate=-27.2]  {$1-p_v^t(j)$};
\draw (348.88,282.21) node [anchor=north west][inner sep=0.75pt]  [font=\Large]  {$p_v^t(j) =1-e^{-\frac{\gamma_t^j}{v\sigma ^{2}P +\sigma _{w}^{2}}}$};
\draw (58.92,90.38) node [anchor=north west][inner sep=0.75pt]  [font=\large]  {$X^{a_{1},j}_t$};
\draw (59.03,226.07) node [anchor=north west][inner sep=0.75pt]  [font=\large]  {$X^{a_{k},j}_t$};
\draw (58.74,135.53) node [anchor=north west][inner sep=0.75pt]  [font=\large]  {$X^{a_{2},j}_t$};
\draw (306.92,240.67) node [anchor=north west][inner sep=0.75pt]    {$k$};
\draw (305.94,110.75) node [anchor=north west][inner sep=0.75pt]    {$1$};
\draw (305.94,88.35) node [anchor=north west][inner sep=0.75pt]    {$0$};
\end{tikzpicture}
}
	\caption{Equivalent characterization of channel in  Fig. 2 corresponding to stage $j$ of multi-stage non-coherent MnAC  with only the active users as inputs. The channel input ${\tilde{\boldsymbol{X}}_t}{(j)}$ is a function of activity status vector $\boldsymbol{\beta}$ and the feedback from previous stages $\{\psi_1,\ldots,\psi_{j-1}\}$ as shown in Fig. 3. }
	\label{fig:eqchnofb}
	\end{figure}
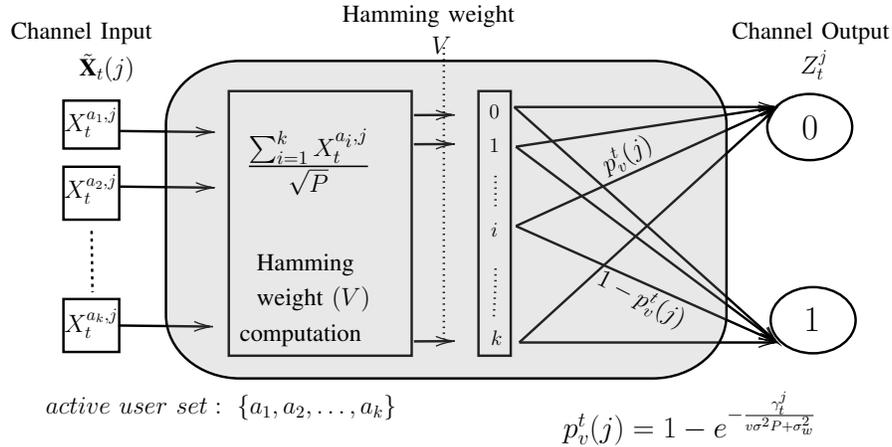

In our prior work \cite{10198448}, we considered a single stage non-coherent $(\ell,k)-$ MnAC without feedback and established that the active device identification problem can be viewed as a  decoding problem in an equivalent point-to-point  channel whose inputs correspond to the active users. Similarly,  in the case of $m$-stage $(\ell,k)-$MnAC, each channel-use $t$ of stage $j$ (see Fig. \ref{fig:qtzn}) can be  translated to an equivalent point to point communication channel  with $\gamma_t^j$ as its detector threshold  as shown in Fig. \ref{fig:eqchnofb}.   

Note that since the inactive users remain `Off', as shown in Fig. \ref{fig:eqchnofb}, the effective input to this equivalent channel during the $t^{th}$  channel-use of stage $j$ is ${\tilde{\boldsymbol{X}}_t}{(j)} =(X^{a_1,j}_{t}, \ldots X^{a_k,j}_{t})$ as defined in Section II.A, where $t \in \{1,\ldots ,n_j\},  j \in \{1,\ldots ,m\}$. This equivalent channel in stage $j$ can be interpreted as a cascade of two channels; During channel-use $t$ of stage $j$, the first channel computes the Hamming weight $V_t^j$ of the effective input ${\tilde{\boldsymbol{X}}_t}{(j)}$ whereas the second channel translates the Hamming weight $V^j_t$ into a binary output $Z^j_t$ depending on the fading statistics $(\sigma^2)$, noise variance $(\sigma_w^2)$ of the wireless channel and the non-coherent detector threshold $\gamma_t^j$. Here onwards, we drop the index $j$ for brevity wherever it is clear from the context.  We showed in \cite{10198448} that the transition probability during stage $j$ of this equivalent channel $p_v^t(j) := p\left(Z_t =z\mid {\tilde{\boldsymbol{X}}_t} ={\boldsymbol{x}}, V_t=v\right)$ is given by 
 \begin{equation}
     p_v^t(j) = (1-z) \left(1-e^{-\frac{\gamma_t^j}{v \sigma^{2}P+\sigma_{w}^{2}}}\right)+ z\left(e^{-\frac{\gamma_t^j}{v \sigma^{2}P+\sigma_{w}^{2}}}\right). \label{pvj}
 \end{equation} Clearly, we have  the Markov chain ${\tilde{\textbf{X}}_t} \rightarrow V_t\rightarrow Z_t, \forall t \in \{1,\ldots,n_j\};\forall j \in \{1,\ldots,m\}$. Exploiting this property, in \cite{10198448}, we derived the maximum rate and minimum user identification cost for the single stage non-coherent $(\ell,k)-$MnAC without feedback, denoted by $C^1$ and $n^1(\ell)$ respectively where the superscript indicates the number of stages (see Def. 4), as summarized in  Theorem \ref{thm1.1} below. 

\begin{theorem}[\cite{10198448}]
For a single stage $(\ell,k)$-MnAC without feedback, the maximum rate of the equivalent point-to-point channel in Fig. \ref{fig:eqchnofb} is \begin{equation} 
    C^1=\max_{(\gamma^1,q^1)} h\Big(E\big[e^{-\frac{\gamma^1}{V  \sigma^{2}P+\sigma_{w}^{2}}}\big]\Big)-E\Big[h\big(e^{-\frac{\gamma^1}{V  \sigma^{2}P+\sigma_{w}^{2}}}\big)\Big].  \label{singlestagecap}
\end{equation}
 where $\sigma^2$ denotes the fading statistics and $\sigma_w^2$ denotes the noise variance.
Furthermore, the minimum  user identification cost $n^{1}(\ell)$ when $k=\Theta(1)$  is given by
 
\begin{equation} 
    n^1(\ell)= \frac{k\log(\ell)}{\max_{(\gamma^1,q^1)} h\Big(E\big[e^{-\frac{\gamma^1}{V  \sigma^{2}P+\sigma_{w}^{2}}}\big]\Big)-E\big[h\big(e^{-\frac{\gamma^1}{V  \sigma^{2}P+\sigma_{w}^{2}}}\big)\Big]}
    \label{eq7}
\end{equation}
where $E(\cdot)$ denotes expectation w.r.t  $V$,  $h(x)=-x \log x-(1-x) \log (1-x)$ is the binary entropy function, $\gamma^1$ denotes the energy detector threshold and $q^1$ denotes the sampling probability used by all active users for preamble generation; i.e., $q_t^{i,1} =q^1, \forall i \in \mathcal{A}$.
\label{thm1.1}
\end{theorem}

Although in \cite{10198448} we considered a fixed threshold $\gamma^1$ for all channel-uses in the single-stage MnAC, the same $C^1$ and $n^1(\ell)$ as in Theorem \ref{thm1.1}  hold for the case where the threshold $\gamma_t^1$ depends on the channel-use index $t$ since we are optimizing over the parameter $\gamma^1$ in (\ref{eq7}). The same rationale applies to the sampling probabilities ${q}_t^{i,1}, \forall i \in \mathcal{A}$. In \cite{10198448}, we established that having the same sampling probability $q_t^{i,1}$ across all active users $ i \in \mathcal{A}$ leads to the maximum rate $C^1$ due to the symmetry of the problem. 

Using ideas from GT \cite{8926588}, in our prior work \cite{10198448}  we showed that the above minimum user identification cost is indeed attainable when the number of active devices, $k$ scales as $k =\Theta(1)$. In addition to this theoretical achievability, we also proposed several practical strategies for  active device identification in non-coherent $(\ell, k)-$MnAC \cite{10198448}. However, we observed a notable gap in the performance of these strategies compared to the theoretical minimum user identification cost. Furthermore, we noted that the relaxation of the exact recovery criterion to a partial recovery one (see Def. 3) improved the user identification costs associated with the practical strategies. Motivated by this observation, in Sections IV and V, we consider a practical multi-stage user identification approach with feedback for active device identification in which we aim to successively refine the partial estimates through multiple stages, eventually leading to an exact recovery after the final stage of processing. Before we discuss these practical approaches, we characterize the theoretical minimum multi-stage user identification cost for an $m$-stage non-coherent $(\ell,k)-$MnAC in the remainder of this section.

 \subsection{Multi-stage minimum user identification cost}

 For the multi-stage setting, note that the equivalent channel in Fig. \ref{fig:eqchnofb} represented as $p_v^t(j)$ in (\ref{pvj}) where $V_t^j=v$ is the Hamming weight  of  ${\tilde{\boldsymbol{X}}_t}{(j)}$ is memoryless both within a stage and across stages since the distribution of the channel output ${Z}_t$ for any channel-use $t$ depends only on the current input  ${\tilde{\textbf{X}}_t}$ and is conditionally independent of previous channel inputs and outputs.  Shannon \cite{1056798} proved that for a discrete memoryless channel, even if all the received symbols are fed back instantaneously and noiselessly to the transmitter in a causal fashion, the capacity is same as the no-feedback case. We build on this idea to characterize the minimum user identification cost for the $m$-stage non-coherent $(\ell,k)-$ MnAC with stage-wise feedback as given in Theorem \ref{thm3} below.

 \begin{theorem}
The minimum  user identification cost $n^m(\ell)$ of the $m$-stage non-coherent  $(\ell,k)-MnAC$ with feedback in the $k=\Theta(1)$ regime does not depend on $m$ is given by

\begin{equation*} 
    n^m(\ell)= \frac{k\log(\ell)}{\max_{(\gamma,\mathbf{q})} h\left(E\left[e^{-\frac{\gamma}{V  \sigma^{2}P+\sigma_{w}^{2}}}\right]\right)-E\left[h\left(e^{-\frac{\gamma}{V  \sigma^{2}P+\sigma_{w}^{2}}}\right)\right]}
\end{equation*}
\normalsize where $E(\cdot)$ denotes expectation w.r.t  $ V$ and $h(x)=-x \log x-(1-x) \log (1-x)$ is the binary entropy. Here, the maximization is done jointly over the energy detector threshold $\gamma$ and the sampling probability vector  $\boldsymbol{q} := [q^1, q^2, \ldots, q^m]$ such that $q_t^{i,j} =q^j,  \forall i \in \mathcal{A}$.
\label{thm3}
\end{theorem}

\begin{proof}
Since a single stage non-coherent $(\ell,k)-$MnAC is a special case of $m$-stage non-coherent $(\ell,k)-$MnAC with sampling probability vector $\boldsymbol{q} =[q^1]$, any user identification cost that can be asymptotically achieved using  a single stage can be achieved with $m$-stages by ignoring the feedback. Hence using Theorem \ref{thm1.1}, 
  \begin{equation}
    n^m(\ell) \leq n^1(\ell) = \frac{k\log(\ell)}{\max\limits_{(\gamma, \boldsymbol{q})}h\Big(E\Big[e^{-\frac{\gamma}{V \sigma^{2}P+\sigma_{w}^{2}}}\Big]\Big)-E\Big[h\Big(e^{-\frac{\gamma}{V \sigma^{2}P+\sigma_{w}^{2}}}\Big)\Big]}.
    \label{p1}
\end{equation}

 Now, we  show that $n^m(\ell) \geq  n^1(\ell)$. Without loss of generality, let $\boldsymbol{\beta}$ be uniformly distributed over $\left\{1, \ldots,  {\ell \choose k}\right\}$. Using Fano's inequality and the data-processing inequality, we have 
\begin{equation}
    \begin{aligned}
\log {\ell \choose k}=H(\boldsymbol{\beta}) & =H(\boldsymbol{\beta} \mid \hat{\boldsymbol{\beta}})+I(\boldsymbol{\beta} ; \hat{\boldsymbol{\beta}}) \\
& \leq 1+\mathbb{P}_e^{(\ell)} \log {\ell \choose k}+I\left(\boldsymbol{\beta} ; \boldsymbol{Z}^1,\ldots \boldsymbol{Z}^m\right).
\end{aligned}
\label{ee1}
\end{equation} We can bound $I\left(\boldsymbol{\beta} ; \boldsymbol{Z}^1,\ldots \boldsymbol{Z}^m\right)$ as follows:
\vspace{-1cm}
\begin{subequations}

\begin{align}
    I\left(\boldsymbol{\beta} ; \boldsymbol{Z}^1,\ldots \boldsymbol{Z}^m\right) & = H\left(\boldsymbol{Z}^1,\ldots \boldsymbol{Z}^m\right) - H\left(\boldsymbol{Z}^1,\ldots \boldsymbol{Z}^m \mid \boldsymbol{\beta}\right) \\
    & = H\left(\boldsymbol{Z}^1,\ldots \boldsymbol{Z}^m\right) - \sum_{j=1}^m H\left(\boldsymbol{Z}^j \mid \boldsymbol{Z}^1, \ldots, \boldsymbol{Z}^{j-1}, \boldsymbol{\beta}\right) \\
    & = H\left(\boldsymbol{Z}^1,\ldots \boldsymbol{Z}^m\right) - \sum_{j=1}^m H\left(\boldsymbol{Z}^j \mid \boldsymbol{Z}^1, \ldots, \boldsymbol{Z}^{j-1}, \boldsymbol{\beta}, {\overline{\boldsymbol{X}}(j)}\right) \label{eq:subeq1}\\
    & = H\left(\boldsymbol{Z}^1,\ldots \boldsymbol{Z}^m\right) - \sum_{j=1}^m H\left(\boldsymbol{Z}^j \mid {\overline{\boldsymbol{X}}(j)}\right)\label{eq:subeq2}\\
& \leq \sum_{j=1}^m H\left(\boldsymbol{Z}^j\right)-\sum_{j=1}^m H\left(\boldsymbol{Z}^j \mid {\overline{\boldsymbol{X}}(j)}\right) \\
& =\sum_{j=1}^m I\left({\overline{\boldsymbol{X}}(j)};\boldsymbol{Z}^j\right) \\
& =\sum_{j=1}^m \sum_{t=1}^{n_j} I\left({\textbf{X}_t(j)};{Z}^j_t\right) \\
& \leq  C^1\sum_{j=1}^m n_j, \label{subeq4}
\end{align}

\end{subequations} \noindent where  $C^1$ as given in Theorem 1 represents the maximum rate of the equivalent point-to-point channel corresponding to a single-stage non-coherent $(\ell,k)$-MnAC. In (\ref{eq:subeq1}), we used the fact that ${\overline{\boldsymbol{X}}(j)}$ is a function of $\{\boldsymbol{Z}^1, \ldots, \boldsymbol{Z}^{j-1}\}$ and $\boldsymbol{\beta}$. For (\ref{eq:subeq2}), we used the Markov chain $\{\boldsymbol{\beta},\boldsymbol{Z}^1, \boldsymbol{Z}^2, \ldots, \boldsymbol{Z}^{j-1}\}\rightarrow \overline{\textbf{X}} (j)\rightarrow \boldsymbol{Z}^j$ along with the memoryless nature of the channel. In (\ref{subeq4}), we used the fact that the equivalent channel model established for the single-stage non-coherent $(\ell,k)-$MnAC holds for each stage of the $m$-stage non-coherent $(\ell,k)-$MnAC. Thus, in the $k=\Theta(1)$ regime, $\mathbb{P}_e^{(\ell)}\rightarrow 0$ as $\ell \rightarrow \infty$ suggests that

  \begin{equation}
    \sum_{j=1}^m n_j \geq \frac{k\log(\ell)}{\max\limits_{(\gamma, \boldsymbol{q})}h\Big(E\Big[e^{-\frac{\gamma}{V \sigma^{2}P+\sigma_{w}^{2}}}\Big]\Big)-E\Big[h\Big(e^{-\frac{\gamma}{V \sigma^{2}P+\sigma_{w}^{2}}}\Big)\Big]},
    \label{p2}
\end{equation}

 where we have used $C^1$ in (\ref{singlestagecap}) along with the fact that $\log {\ell \choose k}\approx k \log(\ell)$ in $k =\Theta(1)$ regime. According to Def. 4, 
 $n^m(\ell) \geq \sum_{j=1}^{m}n_j$ for any activity detection scheme for which $\mathbb{P}_e^{(\ell)} \rightarrow 0$. Combining (\ref{p1}) and (\ref{p2}) completes the proof.
  \end{proof}

Theorem \ref{thm3} implies that a multi-stage approach for active device identification does not yield improvements over the single-stage one in terms of the fundamental limits on minimum user identification cost. However,  in practical terms, multi-stage protocols may provide a potential advantage in enhancing the efficiency of active device identification schemes. In the next section, we propose a practical stage-wise adaptive protocol for active device identification in which a partial estimate of the active device set is formed in the initial stage with subsequent refinements in the following stages to achieve exact recovery. 

 \section{Multi-stage active user identification with feedback}

In our multi-stage active user identification scheme, after each stage, the BS infers a partial estimate of the active set of devices while maintaining some false positive and false negative target. Note that since there are potential false positives, the BS individually tests each device in this partial estimate set to confirm their activity status.  The estimate progressively improves across stages and by the end of the final stage, the BS aims for an exact recovery.  Between stages, using a feedback link, the BS communicates to the devices whether they are estimated as active or not. Recall that our theoretical characterization of the minimum user identification cost  in Section III does not include the channel uses required for providing feedback between stages. However, in the analysis of our proposed practical scheme in this section, we examine the overhead associated with the feedback procedure in terms of the required number of feedback channel uses. Specifically, we demonstrate that this feedback overhead is negligible compared to the minimum user identification cost, justifying its omission in our theoretical characterization.

Let $k_j$ denote the number of active devices that remain to be estimated at the beginning of stage $j$. Similarly, $\ell_j$ represents the number of unclassified devices still present in $\mathcal{D}$. Let the set $\hat{{\mathcal{A}}}_{j}$ of size $k_j$ denote the partial estimate of this remaining set of active devices  $\mathcal{A}_j :=\mathcal{A} \setminus \bigcup_{i=0}^{j-1}\hat{\mathcal{A}}_{i}$ inferred by the BS at the end of stage $j$. Note that  $|\mathcal{A}_j|=|\hat{\mathcal{A}}_j|=k_j$. Also, $\ell_j =|\mathcal{D}\setminus\bigcup_{i=1}^{j-1}\hat{{\mathcal{A}}}_{i}|.$ with $\ell_0:=\ell$.  We initialize $\hat{{\mathcal{A}}}_{0} :=\emptyset$.    For notational convenience, without loss of generality, let us assume that $\mathcal{A}=\{1,\ldots,k\}$ is the active device set. Next, we describe the sequence of processing that happens at each stage $j$.
\begin{itemize}
    \item \textbf{Transmission phase}: For each stage $j$, we consider two sets of devices $\mathcal{S}_1^j$ and $\mathcal{S}_2^j$ defined as follows:

    \begin{enumerate}

        \item $\mathcal{S}_1^j:=\hat{{\mathcal{A}}}_{j-1}$ is the estimate set of active devices identified in previous stage $j-1$.
        \item $\mathcal{S}_2^j:=\mathcal{D}\setminus\bigcup_{i=1}^{j-1}\hat{{\mathcal{A}}}_{i}$ denotes the remaining set of unclassified devices in $\mathcal{D}$.
    \end{enumerate}
  For each device $ i \in \mathcal{S}_1^j$, a non-overlapping block of ${n}_{j_1}$ channel-uses denoted by $\mathcal{B}_i^j$ is allocated by the BS at the end of the previous stage $j-1$ using the feedback phase (discussed later) requiring a total of $k_{j-1}n_{j_1}$ channel-uses. The devices in $\mathcal{S}_1^j$ set their sampling probabilities based on the following:
\begin{equation}
  q_{t}^{i,j} =
    \begin{cases}
      1 & \text{if  $i \in \mathcal{A} \bigcap \mathcal{S}_1^j $  and $t \in  \mathcal{B}_i^j$}\\
      
      0 & \text{otherwise}
    \end{cases}     
    \label{s1j}
\end{equation}

  The devices in $\mathcal{S}_2^j$ set their sampling probabilities based on the following:
\begin{equation}
  q_{t}^{i,j} =
    \begin{cases}
      0 & \text{if  $t \leq k_{j-1}{n}_{j_1}$}\\
      
      q^j & \text{if  $i \in \mathcal{A} \bigcap \mathcal{S}_2^j $ and $ k_{j-1}{n}_{j_1}<t\leq n_j$}
    \end{cases}
    \label{s2j}
\end{equation}
Thus, each active device $i$ in $\mathcal{S}_1^j$   transmits an all-one preamble $\boldsymbol{\vec{1}}_{t\in \mathcal{B}_i^j}$  across their allocated ${n}_{j_1}$ channel-uses in the block $\mathcal{B}_i^j$  whereas others remain silent. Thereafter, the active devices in $\mathcal{S}_2^j$ \textit{jointly} transmit the random  binary preamble  based on $Bern(q^j)$ random variable over $n_{j_2}:= n_j -k_{j-1}n_{j_1}$ channel-uses.

\item \textbf{Reception phase}:
 For each device $i \in {\mathcal{S}}_{1}^j$, the BS listens to their orthogonal transmissions over ${n}_{j_1}$ channel-uses to produce a binary output vector $\boldsymbol{Z}^{i,j_1}:=\{{Z}_t^j: t \in  \mathcal{B}_i^j\}$.
Afterward, the BS produces a binary output vector of length $n_{j_2}$ corresponding to the joint transmission by devices in $\mathcal{S}_2^j$ given by $\boldsymbol{Z}^{j_2}:=\{Z_t^j: k_{j-1}n_{j_1} < t \leq n_{j}\}.$  Note that the BS employs threshold-based non-coherent energy detectors  as described in Section II.A with different thresholds for individual and joint transmissions to produce these sets of binary outputs. 

\item \textbf{Estimation phase}:
Upon reception, the BS classifies each device \(i \in{\mathcal{S}}_{1}^j\) as active or inactive through hypothesis testing using the binary output vector ${\boldsymbol{Z}}^{i,j_1}$.  The block length $n_{j_1}$ is chosen such that  ${P}_e^{(n_{j_1})} \rightarrow 0$.  After that, devices in ${\mathcal{S}}_{1}^j$ do not take part in any further stages of active device identification. 
Next, the BS uses the binary output vectors $\boldsymbol{Z}^{j_2}$  to provide a partial estimate \(\hat{\mathcal{A}}_{j}\) of the set of remaining active devices \( \mathcal{A}_{j}\).  This is done while meeting the partial recovery rate \(\eta_j\) (see Def. 3), ensuring that \(\mathbb{P}_{e,\eta_j}^{(\ell)} \rightarrow 0\) as \(\ell \rightarrow \infty\), with \(k\) replaced by $k_j$. The details of individual hypothesis testing and partial recovery algorithm are provided in Section IV-A and Section IV-B respectively.

 \item \textbf{Feedback phase}: After the estimation phase, the BS broadcasts a  feedback signal $\psi_j$ to the users. The feedback signal  informs the devices in $\hat{\mathcal{A}}_{j}$ about their classification as active and allocates them orthogonal resources for individual hypothesis testing in stage $j+1$ to validate their activity status (see (\ref{s1j})). The  feedback signal also serves the purpose of notifying the remaining unclassified devices in $\mathcal{S}_2^j$ that they will continue their joint preamble transmissions after the hypothesis tests of devices in $\hat{\mathcal{A}}_{j}$ during stage $j+1$ (see (\ref{s2j})). A detailed description of the feedback procedure is provided in Section IV-C.
\end{itemize}

In this manner, we iterate through all of the $m$ stages, gradually refining the set of active users with the help of confirmatory tests and feedback. Note that during the final stage $m$, the partial recovery target $\eta_m$ is set to 100 $\%$ to denote the end of the iterative procedure. The complete protocol is depicted in Algorithm \ref{alg1}. We will numerically evaluate the impact of optimally choosing $\eta_m$'s on minimum user identification cost in Section V.

In the following subsections, we will provide a detailed account of the estimation phase within our active device identification framework.  This phase is characterized by two key steps: individual hypothesis testing for validating the activity status of the partial estimate set from previous stage, and partial recovery  of the remaining active devices based on GT decoding techniques. Following that, we will discuss the feedback procedure.

\begin{algorithm}
\small
\caption{Multi-Stage Active Device Identification with Feedback}
\label{alg:multi-stage-identification-fixed-preambles}
\begin{algorithmic}
\State \textbf{Given}:  $\mathbf{X}^{i,j}, \forall i \in\mathcal{D}; \forall j \in \{1,\ldots,m\}-$Set of preambles;  $\mathcal{D}-$Set of users; $k-$Number of active users; $m$ stages.

\State \textbf{Initialization}: 
 $\eta_j$: Partial recovery rate for stage $j$; Estimate of $\mathcal{A}$: $\hat{\mathcal{A}} = \emptyset$. 
\State  Partial estimate produced by stage $j$: $\hat{\mathcal{A}}_{j} = \emptyset, \forall j \in \{0,\ldots, m\}.$

\For{$j = 1$ to $m$}
\State The remaining set of active devices  $\mathcal{A}_j :=\mathcal{A} \setminus \bigcup_{i=0}^{j-1}\hat{\mathcal{A}}_{i}$. 
\State $k_j := |\mathcal{A}_j|$.
    \State \textbf{{Transmission phase}}: 
     \State \textbf{for} $i \in \mathcal{A} \bigcap \mathcal{S}_1^j:$ transmit an all-one preamble $\boldsymbol{\vec{1}}_{t\in \mathcal{B}_i^j}$ of duration ${n}_{j_1}$. \textbf{ end for}
    \State \textbf{for}    $i \in \mathcal{A} \bigcap \mathcal{S}_2^j $: transmit the   $Bern(q^j)$-preamble  over  channel-uses $\{t: k_{j-1}{n}_{j_1}<t\leq n_j\}$.
    \textbf{ end for}
    \State \textbf{{Reception phase}}:
\State \textbf{for}   $t \in \{1,\ldots, n_{j}\}$, 
  Envelope detection at BS:  $Z_{t}^j=1 \text { if } |S_t|^2>\gamma_t^j ; \text { else } Z_{t}=0.$    
    \textbf{end for}
    \State $\boldsymbol{Z}^{i,j_1}:=\{{Z}_t^j: t\in \mathcal{B}_i^j\}$; $\boldsymbol{Z}^{j_2}:=\{Z_t^j: k_{j-1}n_{j_1} < t \leq n_{j}\}.$
    \vspace{-0.2cm}
    \State \textbf{{Estimation  phase}}:
     \For{ $i \in  \mathcal{S}_1^j:$
 Perform hypothesis testing on $\boldsymbol{Z}^{i,j_1}$ (see Section IV-A).
      \State  \textbf{if} user $i$ is confirmed as active}  add user $i$ to $\hat{\mathcal{A}}$;
        \textbf{end if}
    \EndFor
    \State Use $\mathbf{Z}^{j_2}$ to form the partial estimate $\hat{\mathcal{A}}_j$  such that  $\left| {\hat{\mathcal{A}}_j} \setminus \mathcal{A}_j\right| \leq k_{j}\big(1-\frac{\eta_j}{100}\big)$ (see Section IV-B).
    \vspace{-0.2cm}
    \State \textbf{{Feedback phase}}:
 \State Broadcast feedback signal $\psi_j$ to the users (see Section IV-C).
 \State \textbf{if} $j <m,$ Allocate orthogonal resources $ \mathcal{B}_i^{j+1}$ for each user  $i \in \hat{\mathcal{A}}_{j}$ for stage $j+1$ hypothesis testing. \textbf{end if}

\EndFor
\end{algorithmic}
\label{alg1}
\end{algorithm}
\subsection{Individual hypothesis testing for validation of active devices}
 During stage $j$, each device \(i \in{\mathcal{S}}_{1}^j\), classified as active, transmits an all-one preamble $\boldsymbol{\vec{1}}_{t\in \mathcal{B}_i^j}$  across their allocated ${n}_{j_1}$ channel-uses in the block $\mathcal{B}_i^j$ based on (\ref{s1j}). The corresponding binary output vector ${\boldsymbol{Z}}^{i,j_1}$ of length $n_{j_1}$   is used for hypothesis testing to validate their activity status. Let $H_{z_i}$ denote the Hamming weight of ${\boldsymbol{Z}}^{i,j_1}$.   We consider a majority decoding scheme in which if $H_{z_i} >\frac{n_{j_1}}{2}$, we declare device $i$ as active. An error occurs if any device in ${\mathcal{S}}_{1}^j$ is wrongly validated. Let $P_e^{(n_{j_1})}$ denote the probability than an error occurs during the validation phase. Next, we analyze the probability of error $P_e^{(n_{j_1})}$ associated with this majority decoding scheme.

Let $p_{1 \rightarrow 0} := p\left(Z_t^j =0 | i \in \mathcal{A} \bigcap \mathcal{S}_1^j\right), \forall t \in \mathcal{B}_i^j$  denote the transition probability that an `On' symbol transmitted by an active user in the partial estimate set  $\mathcal{S}_1^j$ is detected as an `Off' symbol at the receiver during  channel-use $t$. Similarly, let $p_{0 \rightarrow 1}:=  p\left(Z_t^j =1 | i \in \mathcal{S}_1^j \setminus  \mathcal{A}\right), \forall t \in \mathcal{B}_i^j$ denote the transition probability that an `Off' symbol corresponding to an inactive user in the partial estimate set  $\mathcal{S}_1^j$ is detected as an `On' symbol at the receiver during channel-use $t$. Note that $p_{1 \rightarrow 0}$ and $p_{0 \rightarrow 1}$ depends on the SNR and the threshold employed at the non-coherent detector similar to the transition probabilities in $(\ref{pvj})$ described in Section III-A. Clearly, $H_{Z_i} \sim Bin(n_{j_1}, 1- p_{1 \rightarrow 0}), \forall i \in \mathcal{A} \bigcap \mathcal{S}_1^j$ and $H_{Z_i} \sim Bin(n_{j_1}, p_{0 \rightarrow 1}), \forall i \in \mathcal{S}_1^j \setminus  \mathcal{A}$. Furthermore, using the majority decoding rule, we note that an active device $i \in \mathcal{A} \bigcap \mathcal{S}_1^j$ is decoded erroneously if $H_{Z_i} \leq  \frac{n_{j_1}}{2}$. Using Chernoff concentration bound  \cite{179368}, $ \forall i \in \mathcal{A} \bigcap \mathcal{S}_1^j$, the probability of error is given by
 \begin{equation}
    \mathbb{P}\left[H_{Z_i} \leq \frac{n_{j_1}}{2} \Big|i \in \mathcal{A} \cap \mathcal{S}_1^j\right] \leq e^{-n_{j_1} D_2\left(\frac{1}{2} \| p_{1 \rightarrow 0}\right)}, 
\end{equation}
where  $D_2(.||.)$ represents the binary KL divergence function \cite{10.5555/1146355}. Similarly,  an inactive device $i \in \mathcal{S}_1^j \setminus  \mathcal{A}$ is decoded as active if $H_{Z_i} >  \frac{n_{j_1}}{2}$. Again, Chernoff concentration bound implies, $\forall i \in \mathcal{S}_1^j \setminus  \mathcal{A}$,  the probability of error is given by 
\begin{equation}
    \mathbb{P}\left[H_{Z_i} > \frac{n_{j_1}}{2} \Big|i \in \mathcal{S}_1^j \setminus  \mathcal{A}\right] \leq e^{-n_{j_1} D_2\left(\frac{1}{2} \| p_{0 \rightarrow 1}\right)}.
\end{equation} 
Let $\rho = \arg\max \limits_{r \in \{ p_{0 \rightarrow 1},  p_{1 \rightarrow 0}\}} {e^{-n_{j_1}D_2\left(\frac{1}{2} \| r\right)}}$. Thus, using union bound and $\rho$, the total probability of error ${P}_e^{(n_{j_1})}$ can be upper bounded as:

\begin{equation}
{P}_e^{(n_{j_1})}\leq k_j \cdot e^{-n_{j_1} D_2\left(\frac{1}{2} \| \rho\right).}
\end{equation} Noting that  $D_2\left(\frac{1}{2} \| \rho\right) = \frac{1}{2} \log _{\frac{1}{4 \rho(1-\rho)}}$ , we have 
\begin{equation}
{P}_e^{(n_{j_1})}\leq k_j \cdot e^{- \frac{n_{j_1}}{2} \log _{\frac{1}{4 \rho(1-\rho)}}}.
\label{cher}
\end{equation}
Note that $\rho$ is bounded away from $0$ and $1$ for any given SNR and detector threshold. Therefore, the right-hand side of (\ref{cher}) tends to zero $n_{j_1} \rightarrow \infty$  when $k_j =\Theta(1)$. Hence, during stage $j$, one can choose $n_{j_1}$ as a slowly monotonically increasing function of $\ell$, for example, $n_{j_1} =\Theta(\log(\log(\ell)))$, to achieve ${P}_e^{(n_{j_1})} \rightarrow 0$ as $\ell \rightarrow \infty.$ Thus, the overall probability of error associated with individual hypothesis testing across all $m$ stages can be driven to zero  using a total of  $\sum_{j=1}^{m}n_{j_1}=\Theta(m\log(\log(\ell)))$ channel uses, which for a finite number of stages $m$, scales as $\Theta(\log(\log(\ell)))$. This will be verified through simulations in Section V.

\subsection{Partial Recovery using Belief Propagation}
In our previous work \cite{10198448}, we presented a GT-based decoder for the exact recovery of active devices which achieves the minimum user identification cost presented in Theorem 1. It is straightforward to extend the same approach for a partial recovery of the active devices. However, we noted that this optimal GT-based decoder is practically infeasible as its implementation requires an exhaustive search over all candidate active user sets. Therefore, in \cite{10198448} we proposed several practical strategies with reduced computational complexity based on noisy-combinatorial matching pursuit (NCOMP) and belief propagation (BP) for the exact recovery as well as partial recovery of active devices. Motivated by this, for each stage $j$ of our multi-state active device identification framework, we employ an extension of the partial recovery strategies in our prior work \cite[Section IV.B] {10198448}  to form $\hat{\mathcal{A}}_j$, a partial estimate of  the remaining set of active devices $\mathcal{A}_j$. This is done  while maintaining a partial recovery rate $\eta_j$ such that  $\left| {\hat{\mathcal{A}}_j} \setminus \mathcal{A}_j\right|  \leq k_j\big(1-\frac{\eta_j}{100}\big)$. Though our proposed framework can incorporate any GT decoding algorithm, in this paper we 
 primarily focus on BP-based algorithms since they are typically superior to NCOMP algorithms and can perform close to the theoretical minimum user identification costs \cite{8926588,10198448}. Nevertheless,  in our numerical simulations in Section V, along with BP algorithm, we also study the performance of NCOMP algorithm for completeness.

For BP based partial recovery, we  consider an approximate bitwise-MAP detection achieved through message passing on a bipartite graph. This graph has all the device nodes from $\mathcal{S}_2^j$ on one side and the binary energy measurements $\{{Z_t}^{j}: t \in \{k_{j-1}n_{j_1},\ldots, n_j\}\}$ on the other side. The priors $\mathbb{P}(\beta_i),\forall i \in \mathcal{S}_2^j$ are initialized to \(\frac{{k_j}}{|\mathcal{S}_2^j|}\) and are updated through message passing in subsequent iterations until convergence. Let \(\mathcal{N}_1(i): =\{t :X_t^i = 1\}\), $\forall$ \(i \in \mathcal{S}_2^j\), and \(\mathcal{N}_2(t): =\{i :X_t^i = 1\}\), $\forall$ \(t \in \{k_{j-1}n_{j_1},\ldots, n_j\}\), denote the neighbor nodes of device \(i\) and channel-use \(t\) in the graph, respectively. Using the standard BP approach \cite{5707018}, the messages \(m_{i \rightarrow t}^{(r)}\) from device node \(i\) to channel-use node \(t\) and the messages \(\hat{m}_{t \rightarrow i}^{(r)}\) from channel-use node \(t\) to device node \(i\) during \(r^{th}\) iteration can be written as:

\begin{equation}
    \begin{gathered}m_{i \rightarrow t}^{(r+1)}\left(\beta_i\right) \propto\left(\frac{{k_j}}{|\mathcal{S}_2^j|} \mathbbm{1}_{\left\{\beta_i=1\right\}}+\left(1-\frac{{k_j}}{|\mathcal{S}_2^j|}\right) \mathbbm{1}_{\left\{\beta_i=0\right\}}\right) \prod_{t^{\prime} \in \mathcal{N}_1(i) \backslash\{t\}} \widehat{m}_{t^{\prime} \rightarrow i}^{(r)}\left(\beta_i\right) \\ \widehat{m}_{t \rightarrow i}^{(r+1)}\left(\beta_i\right) \propto \sum_{\sim \beta_i} \mathbb{P}\left(Z_t \mid \boldsymbol{\beta}\right) \prod_{i^{\prime} \in \mathcal{N}_2(t) \backslash\{i\}} m_{i^{\prime} \rightarrow t}^{(r)}\left(\beta_i\right),\end{gathered}
    \label{[msgp}
\end{equation}
where $\mathbbm{1}_{\left\{\beta_i=0\right\}}$ and $\mathbbm{1}_{\left\{\beta_i=1\right\}}$ denote the indicator random variables corresponding to the events ${\left\{\beta_i=0\right\}}$ and ${\left\{\beta_i=1\right\}}$ respectively.  After several message passing iterations, we expect to reach a sufficiently accurate estimation of the posterior distribution $\mathbb{P}\left(\beta_i \mid \mathbf{Z}^{j_2}\right)$ at each device node in $\mathcal{D}_j$. We employ  BP with Soft Thresholding (BP-ST) in which we classify the ${k_j}$ users with the highest marginals as active. 

\subsection{Feedback procedure}
After the estimation phase in stage $j$, the BS uses a feedback signal $\psi_j$  to inform  
 the users in $\hat{\mathcal{A}}_j$ that they have been identified as active. Furthermore, $\psi_j$   assigns each of the $k_j$  users in  $\hat{\mathcal{A}}_j$, an index indicating their designated order of transmission as reflected in $ \mathcal{B}_i^j$ in  (\ref{s1j}) for  individual hypothesis testing in stage $j+1$ for validating their activity status.  In addition, $\psi_j$   also notifies the users in $\mathcal{S}_2^j\setminus \hat{\mathcal{A}}_j$ that they can proceed with their joint preamble transmissions in stage $j+1$ once all the $k_j$ users in $\hat{\mathcal{A}}_j$  complete their hypothesis testing phases. 
 
 A straightforward approach to provide this feedback would involve assigning a unique index to each of the $\ell_j$ users. Consequently, the feedback signal $\psi_j$ would be an ordered list of the indices corresponding to the $k_j$ users within $\hat{\mathcal{A}}_j$, wherein the order indicates their transmission sequence. Clearly, the overall feedback overhead for this strategy amounts to $n_{fb} = O\left( \log{\ell_j \choose 
 k_j}\right)$ bits. However, such a strategy incurs considerable overhead when compared with the minimal user identification cost $n^m(\ell)$ in Theorem \ref{thm3}, rendering it impractical.

 In the rest of this section, we propose a feedback procedure that can achieve the aforementioned goals of feedback without overwhelming the channel resources. The key to tackling this challenge  is the observation that an active device  monitoring the feedback link is only interested in knowing whether the transmitted feedback contains the identifier of interest for them. If its identifier is present, the active device infers that it will participate in individual hypothesis testing in the next stage. Conversely, if its identifier is absent, the active device concludes that it has not yet been classified as active and will engage in joint preamble transmission in the next stage.
 
Let  $n_{fb}^j$  denote the number of channel-uses required for the BS to provide feedback during stage $j$. We use $n_{fb}:=\sum_{j=1}^{m}n_{fb}^j$ to denote the overall feedback overhead. For ease of analysis, we assume that the feedback link between the BS and the devices is a  broadcast channel of capacity $C_{fb}$ to each user.

The BS assigns each device $i \in \mathcal{D}$ a predetermined feedback codeword $\boldsymbol{W}_i$ of length $\tilde{n}_j$. The  feedback signal $\boldsymbol{\psi}_j$ is an ordered sequence of the feedback codewords corresponding to the users in $\hat{\mathcal{A}}_j$, i.e., $\boldsymbol{\psi_j} := (\boldsymbol{W})_{i \in \hat{\mathcal{A}}_j}$, where the order uniquely determines the orthogonal block of channel uses allocated to each device in $\hat{\mathcal{A}}_j$ for individual hypothesis testing in stage $j+1$. Each device observes a noisy version of $\boldsymbol{\psi_j}$ due to the effects of the feedback channel and  conducts hypothesis tests over blocks of length $\tilde{n}_j$  to determine if its assigned feedback codeword is present and its position in the ordered sequence $\boldsymbol{\psi}_j$. If its feedback codeword is found, the active device concludes that the BS has correctly classified it as active and uses the position to identify their allocated orthogonal block of channel-uses for hypothesis testing during the next stage. Type-I errors (false negatives) occur when an active device incorrectly infers that its codeword missing in $\boldsymbol{\psi}_j$, while type-II errors (false positives) occur when it incorrectly detects its codeword in $\boldsymbol{\psi}_j$. 

This problem is closely related to the identification via channels problem  introduced by Ahlswede and Dueck \cite{42172,9656742} where the receiver is only interested in knowing whether a particular message was sent, but the transmitter does not know which message is of interest to the receiver. The number  of messages that can be reliably transmitted utilizing an identification code is known to be doubly exponential in the blocklength. Specifically, let $\lambda_1$ and $\lambda_2$ denote the upper bounds to  type-I and type-II error probabilities. It has been shown in \cite{42172} that given any $\lambda_1, \lambda_2, \epsilon>0$, there exist identification codewords such that
\begin{equation}
    \frac{1}{\tilde{n}_j} \log \left(\log \left(\ell_{j}\right)\right)>C_{fb}-\epsilon
\end{equation} Clearly, the number of users grows doubly-exponential  with the blocklength. Thus, for each user   in \( \hat{\mathcal{A}}_j \), an identification codeword of length $\tilde{n}_j < \frac{1}{(C_{fb}-\epsilon)}\log \left(\log \left(\ell_{j}\right)\right)$ can be used as the feedback codeword. In \cite{179339}, an explicit construction of binary constant weight codes that can achieve the identification capacity is provided. These codes are three-layer concatenated codes, where the inner code is a pulse position modulated code and the outer codes are Reed-Solomon codes. Overall,  the feedback overhead for stage $j$ amounts to $n_{fb}^j =(\frac{k_j}{C_{fb}})\log \left(\log \left(\ell_{j}\right)\right)(1+o(1))$ which leads to a total feedback overhead given by 
\begin{equation}
    n_{fb}=\sum_{j=1}^{m}\left(\frac{k_j}{C_{fb}}\right)\log \left(\log \left(\ell_{j}\right)\right)(1+o(1))\\
    =\Theta(\log(\log(\ell))).
    \label{nfbval}
\end{equation} 
In our numerical analysis in Section V, we will use (\ref{nfbval}) to compute the feedback overhead. For $C_{fb}$, we will assume a Gaussian channel with  SNR in the feedback link identical to that of the uplink channel. This is a conservative assumption, considering that the BS typically has a higher power budget and may even enhance the SNR. 

In Section IV-A, we noted that the overhead  associated with individual hypothesis testing scales as $\Theta(\log(\log(\ell)))$ similar to the total feedback overhead given in (\ref{nfbval}). In our prior work \cite{10198448}, we showed that the required number of channel-uses for BP based practical strategies  for exact recovery and partial recovery is  close to the minimum user identification which scales as $\Theta(\log(\ell))$.  Thus, we can conclude that the total number of channel-uses for our multi-stage active user identification scheme, including the feedback costs is $O(\log(\ell))$ which is same as the $m$-stage minimum user identification cost given in Theorem 2.  In the next section, we present  extensive numerical simulations studying the performance of different multi-stage practical schemes. Specifically, we focus on two-stage schemes among the multi-stage schemes for brevity.

\section{Numerical Simulations}

First, we consider the performance of single-stage and two-stage schemes for active user identification in a $(1000,20)-$MnAC setting for SNRs ranging from $0$ to $16$ dB. For the two-stage approach, we analyze the performance of various combinations of NCOMP and BP algorithms. The partial recovery rates were set at $\eta_1 = 75\%$ and $\eta_2 = 100\%$  allowing for at most five false positives in stage-1 and exact recovery in stage-2. After the first stage of joint preamble transmission, the BS forms a partial estimate using BP or NCOMP algorithm and provides feedback to specify which devices are classified as active. In the second stage, each of the $20$ device in the partial estimate set undergoes individual hypothesis testing to validate their activity status. Thereafter, a second stage of joint preamble transmission or equivalently, GT, is performed to exactly recover the remaining active devices. 

The required number of channel-uses for active device identification while employing the above two-stage scheme is illustrated in Fig. \ref{figurea}. The theoretical minimum user identification cost as well as the single-stage BP based user identification cost is also included in Fig. \ref{figurea} for reference. Since the minimum user identification cost does not take into account the feedback channel-uses, the overhead associated with feedback is not considered in Fig. \ref{figurea}. Fig. \ref{figureb}, which we discuss later, compares all component costs including feedback.  Fig. \ref{figurea} shows that  the BP algorithm generally performs much better than NCOMP, which aligns with the findings in our prior work \cite{10198448}. Furthermore, compared to a single-stage approach, the two-stage BP-BP approach consistently provides a significant reduction in the required number of channel-uses for active user identification across all ranges of SNR. For example, at an SNR of $10$ dB, the single-stage BP scheme requires around 440 channel-uses whereas two-stage BP-BP scheme requires around 350 channel-uses.

\begin{figure}   
	\centering
	\begin{tikzpicture}
  \sbox0{\includegraphics[width=.55\linewidth,height=75mm,trim={1.0cm 0.8cm 0 0},clip]{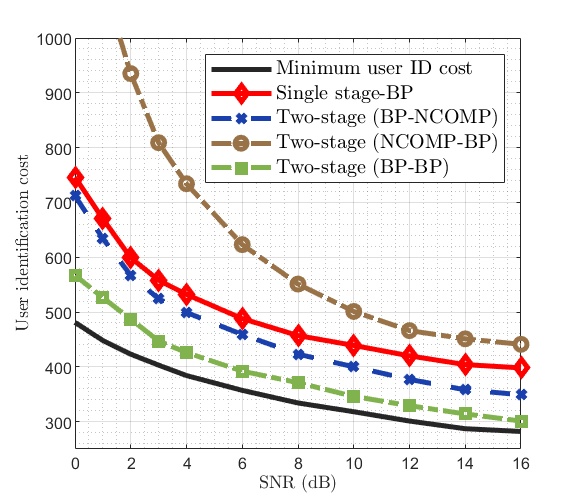}}
  \node[above right,inner sep=0pt] at (0,0)  {\usebox{0}};
  \node[black] at (0.5\wd0,-0.06\ht0) {{SNR (dB)}};
  \node[black,rotate=90] at (-0.04\wd0,0.5\ht0) {{Number of channel-uses $n$}};
\end{tikzpicture}
	\setlength{\abovecaptionskip}{-5pt} 
 \setlength{\belowcaptionskip}{-25pt} 
	\caption{ $(1000,20)$-MnAC for various SNRs with partial recovery rates $\eta_1 = 75\%$ and $\eta_2 = 100\%$ for the two-stage approach. }
\label{figurea}
\end{figure}

 In  Fig. \ref{figureb}, we  consider the same setting as Fig. \ref{figurea} and demonstrate the various costs  associated with the different phases of the proposed two-stage BP-BP user identification scheme including the feedback cost. The different phases and its associated costs are as below.
 \begin{itemize}
     \item  BP-based partial recovery cost: This is the number of channel-uses required to form the partial estimate set using BP in stage 1. In our simulations, we assumed  a partial recovery rate $\eta_1 = 75\%$.
     \item Feedback overhead: The number of feedback channel-uses required to indicate each device if they belong to the set $\hat{\mathcal{A}}_1$. We use (\ref{nfbval}) to compute the feedback overhead $n_{fb}$ assuming a  feedback link over a Gaussian channel with the same average SNR as in the uplink.

     \item  BP-based exact recovery cost: This component corresponds to the number of channel-uses required for the exact recovery of the remaining  up to $25\%$ of the active devices using BP algorithm in stage 2. 
     
     \item Individual hypothesis testing overhead: Each device in  the partial estimate set in stage 1 is subject to individual hypothesis testing for validating their activity status. This is simulated using a majority decoding scheme described in Section IV-A.
 \end{itemize}

 \begin{figure}   
	\centering
	\begin{tikzpicture}
  \sbox0{\includegraphics[width=.65\linewidth,height=75mm,trim={1.5cm 0.8cm 0 1.08cm},clip]{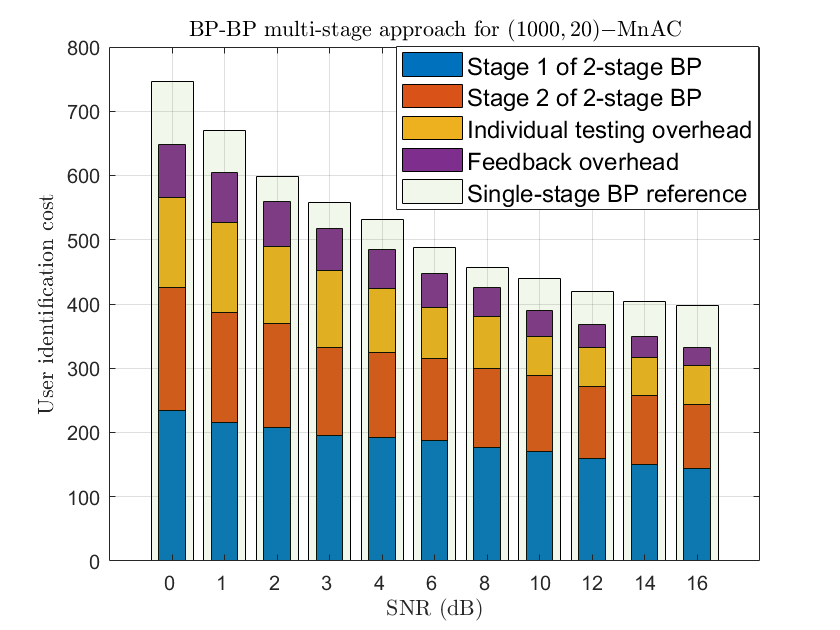}}
  \node[above right,inner sep=0pt] at (0,0)  {\usebox{0}};
  \node[black] at (0.5\wd0,-0.06\ht0) {{SNR (dB)}};
  \node[black,rotate=90] at (-0.04\wd0,0.5\ht0) {{Number of channel-uses $n$}};
\end{tikzpicture}
	\setlength{\abovecaptionskip}{-5pt} 
 \setlength{\belowcaptionskip}{-25pt} 
	\caption{Component costs for BP-BP two-stage user identification scheme in $(1000,20)$-MnAC for various SNRs with partial recovery rates $\eta_1 = 75\%$ and $\eta_2 = 100\%$.}
	\label{figureb}
	\end{figure}

As the SNR increases, the number of channel-uses required for all cost components decrease, gradually reaching saturation at high SNRs. The key observation is that feedback overhead is minimal and is the least costly component among all phases for a given SNR. Specifically, as illustrated in Fig. \ref{figureb}, for moderate to high SNRs, the feedback overhead associated with our two-stage active device identification scheme is in the order of several tens of channel-uses. Since standard LTE or 5G feedback channels typically have  tens to hundreds of channel-uses, incorporating the feedback from our scheme into existing standard protocols is straightforward and can be accommodated within the existing system resources. Fig. \ref{figureb} also reveals that the two-stage BP consistently outperforms single-stage BP for all SNRs even after accounting for feedback overhead.

Next, we study the effect of varying the partial recovery rate of first stage $(\eta_1)$ on the user identification cost. For this, we consider a $(2000, 50)$-MnAC at $10$ dB SNR and vary $\eta_1$ from $90\%$ to $50\%$ as shown in Fig. \ref{figurec}. One key observation is that if the partial recovery rate of stage-1 is too high $(\geq 90\%)$ or too low $(\leq 50\%)$, the two-stage approach seems inefficient compared to the single-stage BP approach. Furthermore, it can be noted that for the above $(2000, 50)$-MnAC at $10$ dB SNR, the most efficient two-stage approach has its partial recovery rate set at $\eta_1 = 70\%$ which amounts to up to 15 false positives in the partial estimate of stage 1.

 \begin{figure}   
	\centering
	\begin{tikzpicture}
  \sbox0{\includegraphics[width=.62\linewidth,height=70mm,trim={1.2cm .8cm 0 1.07cm},clip]{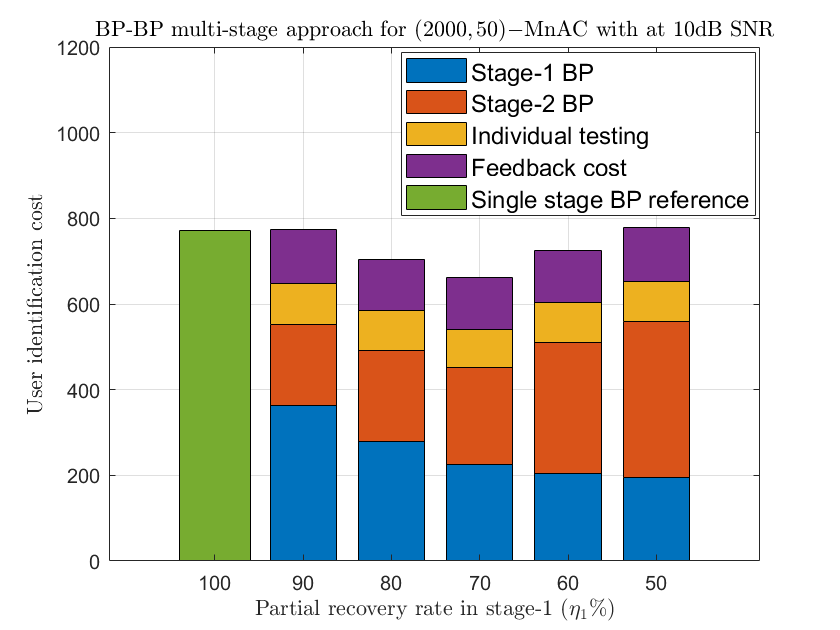}}
  \node[above right,inner sep=0pt] at (0,0)  {\usebox{0}};
  \node[black] at (0.5\wd0,-0.06\ht0) {{Partial recovery rate  in stage-1 ($\eta_1\%$)}};
  \node[black,rotate=90] at (-0.04\wd0,0.5\ht0) {{Number of channel-uses $n$}};
\end{tikzpicture}
	\setlength{\abovecaptionskip}{-5pt}
 \setlength{\belowcaptionskip}{-25pt} 
	\caption{ $(2000,50)$-MnAC for various partial recovery rates $(\eta_1)$ at 10 dB SNR. }
	\label{figurec}
	\end{figure}
\begin{figure}   
	\centering
	\begin{tikzpicture}
  \sbox0{\includegraphics[width=.63\linewidth,height=70mm,trim={1.2cm .8cm 0 1.07cm},clip]{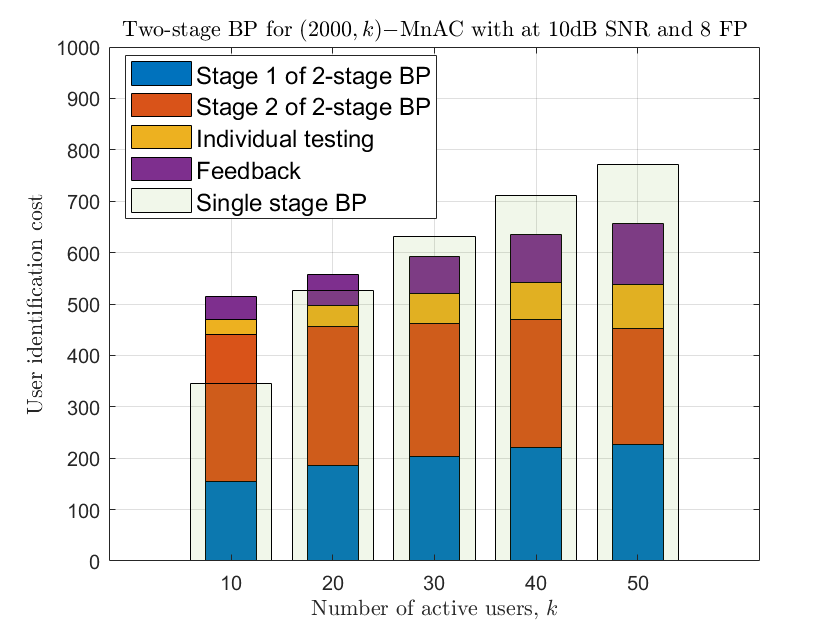}}
  \node[above right,inner sep=0pt] at (0,0)  {\usebox{0}};
  \node[black] at (0.5\wd0,-0.06\ht0) {{Number of active devices, $k$}};
  \node[black,rotate=90] at (-0.04\wd0,0.5\ht0) {{Number of channel-uses $n$}};
\end{tikzpicture}
	\setlength{\abovecaptionskip}{-5pt}
 \setlength{\belowcaptionskip}{-25pt} 
	\caption{ Single-stage BP vs  two-stage BP for $(2000,k)$-MnAC for various values of $k$ at 10 dB SNR with partial recovery rates  $\eta_1 =70\%$ and $\eta_2 =100\%$. }
	\label{figured}
	\end{figure}

 We also analyze our single-stage and two-stage active device identification strategies using BP in a $(2000,k)$-MnAC with the number of active devices $k$ varying from $10$ to $50$ at an SNR of $10$ dB as shown in Fig. \ref{figured}. We observe that for smaller values of $k$, the single-stage approach incurs a lower user identification cost. Essentially, if the problem is sufficiently sparse, there is no necessity to partition it into multiple stages. Conversely, for larger values of $k$, say around $k =30$ and above, the user identification cost of the two-stage BP strategy exhibits significantly lower user identification costs compared to the single-stage BP approach. This suggests that for denser problems, employing a two-stage approach to decompose the problem into a sparse problem can yield savings in terms of the required number of channel-uses for active device identification. In addition, we have numerically analyzed the performance of multi-stage schemes with the number of stages $m$ greater than 2. In general, we observed that the number of stages can be optimized to minimize the number of channel-uses required for user identification. However, due to space constraints, the detailed analysis is omitted here.
 
\section{Conclusions}
In this paper, we have extended our single-stage protocol in \cite{10198448} to propose a multi-stage active device identification framework for non-coherent $(\ell,k)-$MnAC with feedback. Our framework leverages on-off preamble transmission at the device side and threshold-based binary energy detection at the BS side.  We have demonstrated that the active device identification problem in non-coherent $(\ell,k)-$MnAC with feedback can be effectively modeled as a decoding problem in a constrained point-to-point communication channel with feedback. This equivalent modeling  has allowed us to evaluate the performance of our proposed scheme in terms of the minimum user identification cost, representing the minimum number of channel-uses required for reliable active user identification. Specifically, using information theoretic analysis, we have established that the theoretical minimum user identification cost of the non-coherent $(\ell,k)$-MnAC with feedback matches the minimum user identification cost of the single stage non-coherent $(\ell,k)$-MnAC. We have also proposed a practical multi-stage active user identification scheme in which, after each stage, the BS infers a partial estimate of the active set of devices while maintaining some false positive target. The  partial estimate of the active devices  are successively refined across multiple stages using feedback from BS, leading to an exact recovery. We have provided a theoretical characterization of the costs associated with the overheads of feedback and confirmatory hypothesis testing phases within our multi-stage scheme.  We have observed that any efficient scheme for GT-based decoding such as BP algorithm can perform well in a practical multi-stage active device identification framework.  Through several numerical simulations, we have established that BP-based multi-stage strategies can consistently outperform single-stage strategies across a wide range of  SNRs for various choices of partial recovery rates. Thus, our multi-stage active device identification approach demonstrates a lower overall number of channel-uses, or equivalently, reduced resource utilization, despite accounting for feedback overhead and  individual hypothesis testing overhead.

\bibliographystyle{IEEEtranTCOM}

\bibliography{IEEE_TCOM}
\end{document}